\documentclass[siads,final]{siamltex}
\usepackage{amsmath, amssymb, amsfonts}
\usepackage{graphicx,epstopdf,bm}
\usepackage[colorlinks,bookmarks,urlcolor=blue,citecolor=blue,linkcolor=blue]{hyperref}

\newcommand{\pd}[2]{\frac{\partial #1}{\partial #2}}
\newcommand{\pdd}[2]{\frac{\partial^{2} #1}{\partial #2 ^{2}}}

\newcommand{\ave}[1]{\mathbb{E}#1}
\newcommand{\explr}[1]{\exp\left[ #1 \right]}

\newcommand{\chem}[2]{ {{#1\atop\longrightarrow}\atop{\longleftarrow\atop #2}} }

\newcommand{\norm}[1]{\Vert #1 \Vert}
\DeclareMathOperator*{\arginf}{arg\,inf}

\newcommand{\Na}{\text{\tiny Na}}
\newcommand{\K}{\text{\tiny K}}
\newcommand{\rl}{\mathrm{leak}}
\newcommand{\rf}{I_{\rm ion}}
\newcommand{\app}{\mathrm{app}}

\newcommand{\vx}{\mathbf{x}}

\newcommand{\vp}{\bm{p}}

\newcommand{\pv}{p_{v}}
\newcommand{\pw}{p_{w}}
\newcommand{\hgam}{\varphi}
\newcommand{\hgamna}{\tilde{\varphi}}
\newcommand{\Nat}{$\text{Na}^{{+}}$}
\newcommand{\Kt}{$\text{K}^{{+}}$}

\newcommand{\rmp}{\mathrm{p}}
\newcommand{\Dt}{\Delta t}

\newcommand{\act}{S}
\newcommand{\vxa}{\vx_{A}}

\newcommand{\lsim}{\stackrel{\ln}{\simeq}}

\title{Spontaneous excitability in the Morris--Lecar model with ion channel noise\thanks{This work was supported by the Mathematical Biosciences Institute and the National Science Foundation under grant (DMS 0931642)}}
\author{Jay M. Newby\thanks{Mathematical Bioscience Institute, Ohio State University, 1735 Neil Ave. Columbus, OH 43210}}
\begin{document}
\maketitle

\begin{abstract}
Noise induced excitability is studied in type I and II Morris--Lecar neurons subject to constant sub threshold input, where fluctuations arise from sodium and potassium ion channels. Ion channels open and close randomly, creating current fluctuations that can induce spontaneous firing of action potentials. Both noise sources are assumed to be weak so that spontaneous action potentials occur on a longer timescale than ion channel fluctuations. Asymptotic approximations of the stationary density function and most probable paths are developed to understand the role of channel noise in spontaneous excitability. Even though the deterministic dynamical behavior of type I and II action potentials differ, results show that a single mechanism explains how ion channel noise generates spontaneous action potentials.
\end{abstract}

\pagestyle{myheadings}
\thispagestyle{plain}
\markboth{PREPRINT}{\sc Spontaneous excitability in the Morris--Lecar model}

\section{Introduction}
The Morris--Lecar (ML) equations were originally developed as a model of calcium dynamics in muscle fibers of the barnacle {\em Balanus nubilus} \cite{morris81a}.
The ML equations can also be interpreted as simplified version of the Hodgkin--Huxley equations, a well known model of single neuron transmembrane voltage dynamics.
The most widely used simplified version of the Hodgkin--Huxley equations is the so-called FitzHugh--Nagumo equations.
Unlike the simpler FitzHugh--Nagumo equations, ML displays a richer set of dynamics, in particular, several different types of excitability.

The ML equations are given by,
\begin{align}
  \label{eq:28}
  C_{\rm m}\frac{dv}{dt} &= x_{\infty}(v)f_{\Na}(v) + wf_{\K}(v) + f_{\rl}(v) + I_{\app}\\
\nonumber
  \frac{dw}{dt} &= \frac{w_{\infty}(v) - w}{\tau_{w}(v)},
\end{align}
where $v$ is the transmembrane voltage and $w$ represents the fraction of open \Kt~channels.
The functions $f_{i}(v) = g_{i}(v_{i}-v)$ determine the ionic currents.
The fraction of open \Nat~channels is assumed to be an instantaneous function of $v$ with
\begin{equation}
  \label{eq:31}
  x_{\infty}(v) = (1 + \tanh(2(\gamma_{\Na}v + \kappa_{\Na})))/2.
\end{equation}
The steady state fraction and time scale for $w$ are given by
\begin{equation}
  \label{eq:32}
  w_{\infty}(v)=(1 + \tanh(2(\gamma_{\K}v + \kappa_{\K})))/2,\quad \tau_{w}(v) = 2\beta_{\K}\cosh(\gamma_{\K}v+\kappa_{\K}),
\end{equation}
respectively.
(See Appendix \ref{sec:appendix_params} for parameter values.)

The deterministic ML model should be viewed as a mean field limit of a stochastic model that includes the random opening and closings of ion channels.
Single channel opening and closing statistics can be measured experimentally.
The channel variables that modify the ionic conductances represent the fraction of open channels.
The fraction of open channels is a continuous, deterministic quantity if (i) the number of ion channels is taken to be infinite while the conductance of an single channel vanishes or (ii) channels open and close infinitely fast so that the fraction of open channels is an instantaneous function of the voltage.
For the ML model, the potassium channel variable $w$ is obtained by the former while the sodium channel variable $x_{\infty}(v)$ is determined by the latter.
Note that $w$ is a dynamic variable with its own governing equation while $x_{\infty}$ can be viewed as the quasi-steady-state fraction of open sodium channels.

The deterministic ML model can display several different types of excitable behavior.
In every case, there is a single stable fixed point representing the resting voltage and we assume that the applied current $I_{\app}$ is below threshold so that the deterministic system does not exhibit repetitive firing.
Below threshold, only current fluctuations from stochastic ion channels can induce an excitable event.
If current fluctuations push the system over a voltage threshold to the excited state, the voltage undergoes a transient spike called an action potential before returning to the resting voltage.
We consider two situations.
A type I neuron has three fixed points: a stable fixed point corresponding to the resting state, an unstable saddle, and an unstable fixed point corresponding to the excited state.
A type II neuron has one fixed point corresponding to the stable resting state.
Deterministic repetitive firing can occur in type I and II neurons when the input current is increased above threshold \cite{izhikevich00a}.

In most physically relevant cases, the system is close to the deterministic limit so that the ion channel fluctuations are weak compared to the deterministic forces.
In other words, a deterministic trajectory and a stochastic trajectory that share the same initial conditions are likely to remain close over sufficiently small time scales.
This situation is commonly referred to as {\em weak noise}.
Under weak noise conditions, a rare sequence of fluctuations can cause metastable dynamical behavior that the deterministic model cannot describe.
Metastable behavior occurs on long timescales.

A spontaneous excitable event can naturally be separated into two phases: the initiation phase and the excitation phase.
The initiation phase is driven by ion channel fluctuations and is therefore a metastable transition.
The excitation phase begins once fluctuations increase the voltage to a threshold.
Then, the system undergoes a transient increase in voltage before returning to the stable fixed point.
Unlike the initiation phase, the excitation phase is not metastable, instead being driven primarily by deterministic forces.
If we can derive a description of the metastable initiation phase then it can be combined with the deterministic description of the excitation phase to obtain a complete picture of the spontaneous excitable event.

In the weak noise limit, the probability that the process takes a particular path from point A to point B is sharply peaked along a {\em most probable path} (MPP).
A MPP is a statistic, similar to the mode, of a probability distribution functional over the function space of continuous paths.
Although they describe a stochastic process, MPPs themselves are not stochastic.
One can show using {\em large deviations theory} \cite{freidlin12a,feng06a} that the likelihood of deviating from the MPP is an exponentially decreasing function of the magnitude of the deviation.
In other words, stochastic trajectories are highly likely to closely follow the MPP.
In general, there are two kinds of MPPs.
Any deterministic trajectory connecting point A to point B is a MPP.
If there is no deterministic path that connects the two points and the transition is noise induced metastable transition, MPPs build {\em action} and become more improbable as the path gets longer.
The action is a measure of how improbable the MPP is.
By themselves, MPPs provide a good qualitative description of how different metastable transitions occur; they can be thought of as describing noise induced dynamical behavior.

MPPs can also be used to approximate other important properties of the stochastic process.
One can show that there is a nontrivial connection between MPPs that start at the stable fixed point and the stationary probability density function, which describes the relative fraction of time the system spends in different dynamical regimes and determines how rare excitable events are.
Since stochastic trajectories leading from the stable resting voltage to the threshold of an excitable event are described by MPPs, it is no surprise that MPPs determine the asymptotics of the average metastable transition time, also called the mean first passage time (MFPT) or mean exit time.

Several groups have studied stochastic conductance based single neuron models using large deviation theory \cite{berglund06a,khovanov13a,chow96a}.
Until recently, it has only been possible to examine conductance models perturbed by a continuous Markov process.
Channel noise can be approximated by a continuous Markov process, however it is well known that this can generate significant errors for metastable dynamics \cite{newby11b,newby13a}.
Recently, the authors have studied type II excitability in the stochastic ML model with channel noise, deriving MPPs using the WKB method \cite{keener11a,newby13b}.
While the results showed excellent agrement with Monte Carlo simulations, a systematic connection between large deviation theory and the WKB method was not established.

The WKB method is a practical tool used extensively to study metastability in continuous Markov processes and birth-death processes \cite{ludwig75a,dykman96a,maier97a,schuss10a,tel89a}.
The connection between WKB and large deviation theory is well studied for continuous Markov process \cite{ludwig75a} and for birth-death processes \cite{hanggi84a,peliti85a}.
Establishing such a connection for the stochastic ML model is complicated by the presence of fast and slow variables in a stochastic process that has both continuous and discrete elements.
However, due to recent advances in this area \cite{bressloff14a,kifer09a}, a systematic analysis is now possible.
From a practical perspective, establishing a link between large deviation theory and the WKB method facilitates the development of numerical algorithms.
MPPs are computed using the {\em geometric minimum action method} (GMAM) \cite{heymann08a}.
We also develop an {\em ordered upwind method} (OUM) to compute the asymptotic approximation of the stationary density function, based on a similar algorithm for continuous Markov processes \cite{cameron12a}.

The goal of this paper is to develop theory to describe the metastable behavior leading to different types of spontaneous excitation in the stochastic ML model.
In particular, we develop asymptotic approximations of the MPP for a metastable excitable event and the stationary probability density function.

The paper is organized as follows. 
In Section \ref{sec:stoch-model} we introduce the stochastic ML model. 
Then, in Section \ref{sec:mpp} we formulate an approximation of most probable paths, and show how they are connected to the stationary density function.
We show how all relevant quantities can be calculated using the WKB method in Section \ref{sec:wkb}.
Results for type I excitability are presented in Section \ref{sec:sp1}, including spontaneous action potentials and spontaneous bursting.
In Section \ref{sec:sp2}, we present results for the ML model showing type II excitability.

\section{Stochastic model of Morris--Lecar with ion channel noise}
\label{sec:stoch-model}
The stochastic version of the ML model considered here, using simple two-state ion channels, was originally developed in \cite{keener11a,newby13b}.
The voltage equation with $n = 0, 1, \cdots, N$ open \Nat~channels and $m=0, 1, \cdots, M$ open \Kt~channels is
\begin{equation}
  \label{eq:5}
  \frac{dV}{dt} = \rf(v, m, n)  \equiv \frac{n}{N}f_{\Na}(v) + \frac{m}{M}f_{\K}(v) + f_{\rl}(v) + I_{\app}.
\end{equation}
We assume that each channel is either open or closed and switches between each state according to
\begin{equation}
\label{eq:35}
  C\chem{\beta_{i}a_{i}(v)}{\beta_{i}b_{i}(v)}O,\quad i = \mathrm{Na},\;\mathrm{K},
\end{equation}
where the transition rates are $a_{\Na}(v) = e^{4(\gamma_{\Na}v + \kappa_{\Na})}$, $b_{\Na} = 1$, $a_{\K}(v) = e^{\gamma_{\K}v + \kappa_{\K}}$, and $b_{\K}(v) = e^{-\gamma_{\K}v - \kappa_{\K}}$.  
We assume that the \Nat~channels open and close rapidly, so that $1/\beta_{\Na} \ll \tau_{m}$, where $\tau_{m}=C_{\rm m}/g_{\rm L}$ is the membrane time constant.  
Taking $m$ and $n$ in \eqref{eq:5} to be stochastic birth/death processes, we obtain a stochastic hybrid process.
We formulate the process in terms of its probability density function, which satisfies the differential Chapman--Kolmogorov (CK) equation \cite{gardiner83a},
\begin{equation}
\label{eq:34}
  \pd{}{t}\rmp(v, m, n, t) = -\pd{}{v}(\rf(v, m, n)\rmp) + \beta_{\K}\mathbb{L}_{\K}\rmp + \beta_{\Na}\mathbb{L}_{\Na}\rmp.
\end{equation}
The jump operators,
\begin{equation}
\label{eq:12a}
    \mathbb{L}_{\Na} = (\mathbb{E}^{+}_{n}-1)\Omega^{-}_{\Na}(n|v) + (\mathbb{E}^{-}_{n}-1)\Omega^{+}_{\Na}(n|v),
\end{equation}
and
\begin{equation}
\label{eq:12b}
    \mathbb{L}_{\K} = (\mathbb{E}^{+}_{m}-1)\Omega^{-}_{\K}(m|v) + (\mathbb{E}^{-}_{m}-1)\Omega^{+}_{\K}(m|v),
\end{equation}
govern opening/closing of \Nat~and \Kt~channels, respectively, with the jump operator defined by 
\begin{equation}
  \label{eq:29}
  \mathbb{E}^{\pm}_{s}f(s) = f(s\pm 1).
\end{equation}
The transition rates are 
\begin{gather}
  \label{eq:30}
\Omega^{-}_{\Na}(n|v) = n,\quad \Omega^{+}_{\Na}(n|v) = (N-n)a_{\Na}(v),\\ 
\Omega^{-}_{\K}(m|v) = m b_{\K}(v),\quad \Omega^{+}_{\K}(m|v) =(M - m)a_{\K}(v).
\end{gather}

The deterministic system \eqref{eq:28} is recovered in the limit $\beta_{\Na}\to \infty$, $M\to \infty$, and we assume that the limit is taken with $\hgam = \beta_{\Na}/M$ fixed.  
After setting $x = n/N$ and $w = m/M$, the limit yields $x_{\infty}(v) = a_{\Na}(v)/(1+a_{\Na}(v))$ and $w_{\infty}(v) = a_{\K}(v)/(b_{\K}(v)+a_{\K}(v))$, which is consistent with \eqref{eq:28} \cite{keener11a}. 
The parameter $\beta_{\K}$ determines how rapidly the \Kt~channels fluctuate.
Here, we assume that $v$ and $w$ change on the same timescale, with $\tau_{m}\beta_{\K}=O(1)$.

The model has two large parameters, and in order to obtain a single small parameter to carry out a systematic perturbation expansion, we define $\epsilon \ll 1$ such that $\beta_{\Na}^{-1} = \hgamna\epsilon$ and  $M^{-1} = \hgam \epsilon$, with $\hgamna/\tau_{m} = O(1)$ and $\hgam/\tau_{m} = O(1)$. 
(We set $\hgamna = 1$.)
Of course, $N$ could also be a large parameter, but taking the limit $N\to\infty$, $M\to\infty$ yields a different deterministic limit than \eqref{eq:28} (requiring an additional equation for the \Nat~conductance).  We emphasize that our choice of scaling means that the approximation is valid for any choice of $N>0$.

\section{Most probable paths}
\label{sec:mpp}

In this section, we develop theory for calculating asymptotic approximations of MPPs and the stationary density.
We assume that the system is close to deterministic, with a single small parameter $0 < \epsilon \ll 1$ such that in the limit $\epsilon \to 0$, the deterministic system is reached.
First we develop a path description of the process that allows us to define a most probable path.
We also obtain a large deviation principle that established the connection between MPPs and the stationary density function.

The approximations introduced in the section are based on the idea of formulating Laplace's method (a special case of the method of steepest decent) specifically for probability theory \cite{freidlin12a}.
We introduce the notation ``$a\lsim b$'' to represent a logarithmically asymptotic relation, e.g.,
\begin{equation*}
  \label{eq:54}
  f(x; \epsilon) \lsim e^{R(x)/\epsilon} \Rightarrow  \epsilon \log(f(x; \epsilon)) = R(x) + O(\epsilon).
\end{equation*}
Recall that Laplace's method is used to approximate certain integrals.
For example, suppose we have a positive function $F(x)$ for which there is a single minimum at $x_{m}$ such that $F'(x_{m}) = 0$ with $\lim_{x\to \pm \infty}F(x) = \infty$.
Laplace's method yields the approximation
\begin{equation*}
  \int_{-\infty}^{\infty}e^{-F(x)/\epsilon} dx \lsim e^{- F(x_{m})/\epsilon}.
\end{equation*}
The idea is that the integral takes its largest contribution at the point $x_{m}$ where $F(x)$ is at its minimum because the integrand is exponentially decreasing away from $x_{m}$.

To see how this works in the context of probability theory, consider the random variable $X_{\epsilon}$.
We want to approximate the distribution with
\begin{equation}
\label{eq:57}
  {\rm Pr}[X_{\epsilon}\in(x, x+dx) ] \lsim dx\, e^{-L(x)/\epsilon},
\end{equation}
for some function $L(x)$ defined as follows.
Let $p \in \mathbb{R}$.
Using \eqref{eq:57} we notice that
\begin{equation*}
  \ave{\explr{\frac{p X_{\epsilon}}{\epsilon} }} \lsim \int_{-\infty}^{\infty}\explr{\frac{p x - L(x)}{\epsilon} }dx,
\end{equation*}
where $\ave{}$ denotes the expectation or mean.
Laplace's method is based on the idea that the above integral obtains its largest contribution at $H(p) \equiv \sup_{x \in \mathbb{R}}\{ xp - L(x)\}$ so that $\ave{\exp\{p X_{\epsilon}/\epsilon \}} \lsim \exp\{ H(p)/\epsilon\}$.
Hence, we can rewrite $H(p)$ as
\begin{equation}
\label{eq:72}
  H(p) = \lim_{\epsilon \to 0}\epsilon \log \ave{\explr{\frac{p X_{\epsilon}}{\epsilon} }}.
\end{equation}
Assume that $h$ is continuous and convex for all $p\in \mathbb{R}$ with $H(0) = 0$.
Then, the desired function $L(x)$ is related to $H(p)$ through the so called {\em Legendre transform}:
\begin{equation}
  \label{eq:73}
  L(x) = \sup_{p\in\mathbb{R}}\{xp - H(p) \}, \quad H(p) = \sup_{x \in \mathbb{R}}\{ xp - L(x)\}.
\end{equation}
After differentiating $L$ with respect to $x$ and $H$ with respect to $p$, we observe that the maximizers are given implicitly by
\begin{equation*}
  x = H'(p), \quad p = L'(x).
\end{equation*}
The above is well defined if $L(x)$ and $H(p)$ are convex functions.
Hence, the procedure for Laplace's method is to first compute $H$ using \eqref{eq:72}.
If the limit exists then the probability density function can be approximated by \eqref{eq:57}, where $L(x)$ is given by \eqref{eq:73}.
This procedure can be generalized to apply Laplace's method to functional integrals.

For notational convenience, we define $\vx = (v, w)$, treating $w = m/M$ as a continuous variable.
Discretize time with $t_{j} = t_{0} + j \Dt$, $j = 0,\cdots, J+1$.  Let $\{n_{j}\}$ and $\{\vx_{j}\}$ be a discretized path where $n_{j} = n(t_{j})$ and $\vx_{j} = \vx(t_{j})$.
Assume that the end points ($j=0$ and $j=J+1$) of the path are fixed with $n_{J+1} = n$ and $\vx_{J+1} = \vx$.
The probability density function satisfying \eqref{eq:34} can be written in terms of the compounded CK equation, which is a slightly different version of the standard integral form of the Chapman--Kolmogorov (CK) equation.
If the path consists of a single interior point (i.e., $J=1$) then we have the standard CK equation,
\begin{equation}
  \label{eq:124}
  \rmp(n, \vx, t | n_{0}, \vx_{0}, t_{0}) = \sum_{n_{1} = 0}^{N}\int_{-\infty}^{\infty}d\vx_{1}\rmp(n, \vx, t | n_{1}, \vx_{1}, t_{1})\rmp(n_{1}, \vx_{1}, t_{1} | n_{0}, \vx_{0}, t_{0}).
\end{equation}
Compounding over multiple interior points along a path, we have
\begin{equation}
  \label{eq:43}
  \rmp(n, \vx, t | n_{0}, \vx_{0}, t_{0})  =  \sum_{n_{1},\cdots, n_{J}}\int_{-\infty}^{\infty}d\vx_{1}\cdots\int_{-\infty}^{\infty}d\vx_{J}\mathcal{P}[\{n_{j}\}, \{\vx_{j}\}],
\end{equation}
where $\mathcal{P}$ is the joint distribution over the path.
Using the Markov property, the path distribution can be written as the product
\begin{equation}
  \label{eq:129}
  \mathcal{P}[\{n_{j}\}, \{\vx_{j}\}] = \prod_{j=1}^{J+1}\rmp(n_{j}, \vx_{j}, t_{j} | n_{j-1}, \vx_{j-1}, t_{j-1}).
\end{equation}

The goal is to find a useful approximation for the path distribution \eqref{eq:129} that can be used to calculate MPPs.
Formally, we take a continuum limit $J\to \infty$ with $\Dt\to 0$ so that $n(t_{j}) \to n(t)$ and $\vx(t_{j}) \to \vx(t)$.
For $\epsilon \ll 1$, the marginal path distribution (after summing over the $\{n_{j}\}$) can be approximated (see Appendix \ref{sec:app_path}) by
\begin{equation}
  \label{eq:154}
  \mathcal{P}[\vx(t)] \lsim \explr{-\frac{1}{\epsilon}\int_{t_{0}}^{t}L[\vx(s), \vx'(s)]ds},
\end{equation}
where $\vx'(t) = \frac{d\vx}{dt}$.
The {\em Lagrangian} and the {\em Hamiltonian} are related by the Legendre Transform:
\begin{equation}
  \label{eq:48}
  L[\vx, \vx'] = \sup_{\vp \in \mathbb{R}^{2}}\left\{\vp \cdot \vx' - \mathcal{H}(\vx, \vp)\right\}, \quad
  \mathcal{H}(\vx, \vp) = \sup_{\vx'\in \mathbb{R}^{2}}\{ \vp\cdot \vx' - L[\vx, \vx'] \}.
\end{equation}
where $\vp = (\pv, \pw)$ is called the {\em conjugate momentum} which plays an important role as explained below.
For \eqref{eq:48} to be well defined, $L$ must be a convex function of $\vx'$ and $\mathcal{H}$ must be a convex function of $\vp$.
Given the convexity requirements, it follows that the derivatives $\nabla_{\vx'}L[\vx, \vx']$ and $\nabla_{\vp}\mathcal{H}(\vx, \vp)$ are monotonic functions, and the maximizers are given by
\begin{equation}
  \label{eq:53}
  \vx' = \nabla_{\vp}\mathcal{H}(\vx, \vp),  \quad \vp = \nabla_{\vx'}L[\vx, \vx'].
\end{equation}
There is a generalization of \eqref{eq:72} that can be used to compute $\mathcal{H}$, but it is more useful in practice to use the WKB method, as shown below in Section \ref{sec:wkb}.
Two general and essential properties of $\mathcal{H}$ (which are established concretely in Section \ref{sec:wkb}) that we use repeatedly throughout the paper involve its behavior at $\vp = 0$, namely
\begin{equation}
  \label{eq:151}
  \mathcal{H}(\vx, 0) = 0, \quad \nabla_{\vx}\mathcal{H}(\vx, 0) = \frac{d }{dt}\vx_{\rm det},
\end{equation}
where $\frac{d }{dt}\vx_{\rm det}$ is the deterministic dynamics that satisfies the ML equations \eqref{eq:28}.

Define the {\em action} (sometimes called the rate function) as
\begin{equation}
  \label{eq:134}
  \act(t)  \equiv \int_{t_{0}}^{t}L[\vx(s), \vx'(s)]ds.
\end{equation}
It follows from the definition \eqref{eq:48} and \eqref{eq:151} that $S(t)$ is a nondecreasing function, and that $S' = 0$ in the limit $\vx \to \vxa$.
This shows that the action is a measure of how likely a given trajectory is.
As the action increases, the likelihood of observing the trajectory decreases.

In functional integral notation, the marginal density function (after summing over $n$ and $n_{0}$) can be written as
\begin{equation}
  \label{eq:47}
  \rmp(\vx, t | \vx_{0}, t) \lsim \int_{\mathbf{C}_{\vx_{0}}^{\vx}}\mathcal{D}[\vx(t)]\explr{-\frac{1}{\epsilon}S(t)}.
\end{equation}
where $\mathbf{C}_{\vx_{0}}^{\vx}$ is the set of all continuous paths such that $\vx(t_{0}) = \vx_{0}$ and $\vx(t) = \vx$.
Fortunately, we never have to explicitly evaluate \eqref{eq:47}, instead we use Laplace's method to obtain a useful approximation.
A most probable path is defined as the largest contribution to \eqref{eq:47}, 
\begin{equation}
  \label{eq:55}
  \vx_{\rm MP}(t) \equiv \arginf_{\vx(t) \in \mathbf{C}_{\vx_{0}}^{\vx}}\int_{t_{0}}^{t}L[\vx(s), \vx'(s)]ds.
\end{equation}
One can show using the calculus of variations that the MPP defined by \eqref{eq:55} is a solution to the Euler--Lagrange equation,
\begin{equation}
  \label{eq:98}
  \frac{d}{dt}\nabla_{\vx'}L[\vx, \vx'] = \nabla_{\vx}L[\vx, \vx'].
\end{equation}
The above can be equivalently written as
\begin{equation}
  \label{eq:115}
  \frac{d}{dt}\left(\vx'\cdot \nabla_{\vx'}L[\vx, \vx'] - L[\vx, \vx']\right) = 0.
\end{equation}
Using \eqref{eq:48} and \eqref{eq:53}, we have the first integral,
\begin{equation}
  \label{eq:123}
  \vx'\cdot \vp - L[\vx, \vx'] = \mathcal{H}(\vx, \vp) = \text{Const}.
\end{equation}
Hence, $\mathcal{H}$ is constant along a MPP.~ 
Using Hamilton's principle, MPPs satisfy Hamilton's equations,
\begin{equation}
\label{eq:41}
  \vx' = \nabla_{\vp}\mathcal{H}(\vx, \vp), \quad \vp' = -\nabla_{\vx}\mathcal{H}(\vx, \vp).
\end{equation}
The first equation simply follows from \eqref{eq:53}.
The second equation is derived as follows.
Using the second equation in \eqref{eq:53} and \eqref{eq:98}, we have 
\begin{equation}
  \label{eq:126}
  \vp'(t) \equiv \frac{d}{dt}\vp = \frac{d}{dt}\nabla_{\vx'}L[\vx, \vx'] = \nabla_{\vx}L[\vx, \vx'].
\end{equation}
We therefore need to show that $\nabla_{\vx'}L[\vx, \vx'] = -\nabla_{\vx}\mathcal{H}(\vx, \vp)$ on MPPs.
Using \eqref{eq:48} and writing the maximizer $\vp(\vx)$ (i.e., the implicit solution of $\nabla_{\vx}\mathcal{H}(\vx, \vp) = \vx'$) as a function of $\vx$ yields
\begin{equation}
  \label{eq:132}
  \nabla_{\vx}L[\vx, \vx'] = \pd{\vp}{\vx}\cdot \vx' - \pd{\vp}{\vx}\cdot\nabla_{\vp}\mathcal{H}(\vx, \vp(\vx)) - \nabla_{\vx}\mathcal{H}(\vx, \vp(\vx)).
\end{equation}
Then, it follows from the first equation in \eqref{eq:53} that $\nabla_{\vx}L[\vx, \vx'] = - \nabla_{\vx}\mathcal{H}(\vx, \vp(\vx))$.

The relationship between the Hamiltonian and Lagrangian in \eqref{eq:48} provides the connection between MPPs and the stationary density.
If we want to move the starting point to the stable fixed point, we must take the limit $t_{0}\to -\infty$ as $\vx_{0}\to \vxa$, because $\vxa$ is an unstable saddle point of \eqref{eq:41} (even though it is a stable fixed point of the deterministic system \eqref{eq:28}). 
We also have that $\vp \to 0$ as $\vx_{0} \to \vxa$.
Recall that the deterministic system \eqref{eq:28} is recovered from \eqref{eq:41} by setting $\vp = 0$, which means that deterministic flows exist on the stable manifold of the higher dimensional system.
For any time autonomous stochastic process, the stationary density can be written as
\begin{equation}
  \rmp_{\infty}(n, \vx) = \lim_{t_{0}\to -\infty}\rmp(n, \vx, t | n_{0}, \vxa, t_{0}).
\end{equation}
Since $\mathcal{H}$ is constant along a MPP and $\mathcal{H}(\vx, 0) = 0$, we have that $\mathcal{H}(\vx, \vp) = 0$ along every MPP that starts at $\vxa$.
With this in mind, we can approximate the stationary density by taking the limit $t_{0}\to-\infty$ and $\vx_{0}\to\vxa$ in \eqref{eq:47} to get
\begin{equation}
\label{eq:37}
\begin{split}
  \rmp_{\infty}(n, \vx) \lsim  \explr{-\frac{1}{\epsilon}\int_{-\infty}^{t}\vp(s)\cdot \vx'_{\rm MP}(s) ds} 
\equiv \explr{-\frac{1}{\epsilon}W(\vx)}.
\end{split}
\end{equation}
Hence, the conjugate momentum $\vp$ has the alternative definition, $\vp = \nabla_{\vx}W$, where $W$ is called the {\em quasipotential}.
It follows from \eqref{eq:47} and \eqref{eq:37} that the action determines the quasipotential with $W(\vx(t)) = \act(t)$ along the $\mathcal{H}=0$ MPPs.
Notice that if MPPs are deterministic trajectories (with $\vp = 0$) the quasipotential is flat.
This the case at deterministic fixed points.

An alternative way to derive \eqref{eq:41} is to apply the method of characteristics \cite{ockendon03a} to the static Hamilton--Jacobi equation,
\begin{equation}
  \label{eq:66}
  \mathcal{H}(\vx, \nabla W) = 0.
\end{equation}
As shown in the next section, the above equation arises as the leading order problem in the WKB expansion.
The curves $(\vx(t), \vp(t))$ are called {\em characteristics}, and the lower dimensional curves $\vx(t)$ are called {\em characteristic projections}.

Given initial data parameterized by $\theta$, namely $\vx(0) = \vx_{0}(\theta)$ and $\vp(0) = \vp_{0}(\theta)$, the set of characteristics parameterizes the solution surface $W$ to \eqref{eq:66}.
As with any nonlinear scalar PDE, the method of characteristics can break down if characteristic projections cross.
This corresponds to the solution surface folding over on itself, and some additional constraint is necessary to obtain a unique solution.
For example, if two different characteristic projections, $\vx_{1}(t),\vx_{2}(t)$ cross at some time $T>t_{0}$ so that $\vx_{1}(T) = \vx_{2}(T) = \vx$ and $S_{1}(T) <  S_{2}(T)$, then there are two possible values that the quasipotential can take: $W(\vx) = \act_{1}(T)$ or $W(\vx) = \act_{2}(T)$.
One of the key results of the large deviation principle is the additional contraint to resolve a unique solution.

Recall that a MPP becomes less likely as the action increases.
At a given point $\vx$ where two or more characteristic projections cross at time $T$, we designate the MPP to be the characteristic with the smallest corresponding action $\act(T)$ defined by \eqref{eq:134}.
This is known as the {\em least action principle}.
A {\em caustic} is a curve along which each point is the terminus of two or more MPPs that have equal action.
For the above example, we would set $W(\vx) = \act_{1}(T)$ since $S_{1}(T) < S_{2}(T)$.
If instead we have $S_{1}(T) = S_{2}(T)$ then the point $\vx$ is on a caustic.
Even though $\vx_{1}(T) = \vx_{2}(T)$ and $S_{1}(T) = S_{2}(T)$ at a point on a caustic, it does not necessarily follow that $\vp_{1}(T) = \vp_{2}(T)$, which means that the gradient of the quasipotential is discontinuous across the caustic.

The large deviation formulation is not just useful for establishing the connection between the WKB approximation of the stationary density and MPPs.
The variational aspect of the theory can be used to develop numerical methods.
We use two numerical methods to compute MPPs and the quasipotential.
MPPs are computed using the geometric minimum action method (GMAM) \cite{heymann08a}, which yields a numerical approximation of $\vx(t), \vp(t), \act(t)$ between two given points. (Note that the approximation is discretized by arclength in $(v, w)$ and not by time).

It is often desirable to have a more global view of the quasipotential.
Instead of approximating solutions to \eqref{eq:41} like the GMAM, a numerical finite diference scheme can be used to approximate the solution to the static Hamilton--Jacobi equation \eqref{eq:66}.
To compute the quasipotential at points on a discrete grid in the $\vx$ plane, we have developed an {\em ordered upwind method} (OUM), the details of which can be found in Appendix \ref{sec:app_oum}.
One of the key advantages of the OUM is that it naturally resolves caustics by updating grid points in order of increasing quasipotential.

\subsection{WKB approximation}
\label{sec:wkb}
The WKB method is a well known means of approximating the stationary density.
It is also a practical way to calculate the Hamiltonian $\mathcal{H}$.
By conditioning on the number of open \Nat~channels $n$, we can decompose the stationary density with $\rmp_{s}(n, \vx) = r(n| \vx)u(\vx)$, where $u(\vx) \equiv \sum_{n=0}^{N}\rmp_{s}(n, \vx)$ is the marginal density function and $r(n| \vx)$ is the steady-state conditional distribution for $n$ given $\vx$.

Motivated by \eqref{eq:37}, we assume the stationary solution to \eqref{eq:34} has the form
\begin{equation}
  \label{eq:40}
  \rmp_{s}(n, \vx) \sim \mathcal{N} \left[r_{0}(n | \vx) + \epsilon r_{1}(n| \vx)\right]k(\vx)\explr{-\frac{1}{\epsilon}W(\vx)},
\end{equation}
where $\mathcal{N}$ is a normalization constant, $W(\vx)$ is the quasipotential, and $k(\vx)$ is called the {\em pre exponential factor}.
It follows that
\begin{equation}
  \label{eq:148}
  r(n| \vx) \sim r_{0}(n| \vx) + \epsilon r_{1}(n|\vx), \quad u(\vx) \sim \mathcal{N}k(\vx)\explr{-\frac{1}{\epsilon}W(\vx)}.
\end{equation}
At leading order, the WKB expansion determines $r_{0}$ and $\mathcal{H}$, which determines the quasipotential $W$.
The pre exponential factor is determined at higher order (see Appendix \ref{sec:prefactor}).

Substituting \eqref{eq:40} into \eqref{eq:34} and collecting leading order terms yields
\begin{equation}
  \label{eq:44}
  \left[\frac{1}{\hgamna}\mathbb{L}_{\Na} + \pv \diag{\rf(\vx, n)} + h(\vx, \pw)\mathbb{I}\right]r_{0}(n | \vx) = 0,
\end{equation}
where
\begin{equation}
  \label{eq:45}
       h(\vx, \pw) = \frac{\beta_{\K}}{\hgam}\sum_{j=\pm}(e^{-j\hgam \pw}-1)\Omega^{\pm}_{\K}(Mw|v)/M.
\end{equation}
Hence, $r_{0}$ is the appropriately normalized nullvector of the discrete operator,
\begin{equation}
  \label{eq:46}
  \mathbb{M}(\vp) \equiv \frac{1}{\hgamna}\mathbb{L}_{\Na} + \pv \diag{\rf(\vx, n)} + h(\vx, \pw)\mathbb{I},
\end{equation}
which can be rewritten as a matrix.
In order for the WKB solution to be non negative, the nullvector must also be non negative.
Using the Perron--Frobenius Theorem, one can show (see Lemma \ref{lem:PF} in Appendix \ref{sec:app_pathequiv}) that $\mathbb{M}$ has a unique positive eigenvector corresponding to a real, simple eigenvalue that is greater than the real part of all other eigenvalues.
One can also show that the Perron eigenvalue (see Appendix \ref{sec:app_path}) is the Hamiltonian $\mathcal{H}$ that defines the Lagrangian in \eqref{eq:48}, which means that characteristics are MPPs.
Hence, we rewrite \eqref{eq:44} as the eigenvalue problem,
\begin{equation}
  \label{eq:51}
  \left[\mathbb{M}(\vp) - \mathcal{H}(\vx, \vp)\mathbb{I}\right]q(n, \vp) = 0,
\end{equation}
with $q(n, \vp) > 0$ for all $n$ and $\vp$.
Equation \eqref{eq:51} is equivalent to \eqref{eq:44} when $\mathcal{H}(\vx, \vp) = 0$, which is the Hamilton--Jacobi equation \eqref{eq:66} derived in the previous section!
Given $\vp(\vx)$ such that $\mathcal{H}(\vx, \vp(\vx)) = 0$, we have that $r_{0}(n | \vx) = q(n, \vp(\vx))$.

To calculate $\mathcal{H}$ we use the anzatz,
\begin{equation}
  \label{eq:49}
  q(n) = A^{n}/(n!(N-n)!),
\end{equation}
where $A$ is an unknown that must be determine self consistently.
Substituting \eqref{eq:49} into \eqref{eq:51} yields
\begin{multline}
  \label{eq:50}
  \left[-A^{2} - A\left(1 - a - \frac{\hgamna\pv}{N}f_{\Na}\right)+ a_{\Na}\right]n  \\
+ \frac{A}{N}\left[A - a_{\Na} + \frac{\hgamna}{N}(\pv g + h - \mathcal{H})\right] = 0.
\end{multline}
Setting 
\begin{equation}
  \label{eq:52}
    A  = a_{\Na}(v) - \frac{\hgamna}{N}(\pv g(\vx) + h(\pw) - \mathcal{H}(\vx, \vp)),
\end{equation}
eliminates the $n$ independent term.
The remaining $n$ dependent term becomes
\begin{multline}
  \label{eq:117}
  (h - \mathcal{H})^{2}  + \left[(2g + f_{\Na})\pv - \frac{N}{\hgamna (1-x_{\infty})} \right] (h - \mathcal{H}) \\
+  (f_{\Na} + g)g\pv^{2} - \frac{N(x_{\infty}f_{\Na} + g)}{\hgamna(1-x_{\infty})}\pv = 0,
\end{multline}
where we use $x_{\infty} = a_{\Na}/(1 + a_{\Na})$.
Hence, the solvability condition \eqref{eq:50} is satisfied by setting
\begin{equation}
\label{eq:63}
  \mathcal{H}(\vx, \vp) = h(\vx, \pw)  + \frac{1}{2}\left[\frac{N}{\hgamna (1-x_{\infty})}  - (2g + f_{\Na})\pv\right] + \frac{1}{2}z(\pv)^{\frac{1}{2}},
\end{equation}
where
\begin{equation}
  \label{eq:65}
  h(v, w, \pw) = \frac{\beta_{\K}}{\hgam}\sum_{j=\pm}(e^{-j\hgam \pw}-1)\Omega^{\pm}_{\K}(Mw|v)/M,
\end{equation}
and
\begin{equation}
\label{eq:64}
\begin{split}
    z(\pv) &\equiv \left[(2g + f_{\Na})\pv - \frac{N}{\hgamna (1-x_{\infty})} \right]^{2} \\
&\quad- 4\left[ (f_{\Na} + g)g\pv^{2} - \frac{N(x_{\infty}f_{\Na} + g)}{\hgamna(1-x_{\infty})}\pv\right].
\end{split}
\end{equation}

It is worth pausing here to mention that for general stochastic hybrid weak noise problems, it is not always possible to solve the characteristic equation to obtain the Hamiltonian.
In Ref.~\cite{newby13b}, a different Hamiltonian was derived and used to generate MPPs by setting $\mathcal{H} = 0$ in the characteristic equation \eqref{eq:117}.
An alternative to defining the Hamiltonian as the Perron eigenvalue is to define it as the determinant of the matrix obtained from the WKB method (e.g., \eqref{eq:46}).
Call the alternative Hamiltonian $\widehat{{\cal H}}$.
While it immediately follows that this alternative formulation defines the same gradient $\vp$ of the quasipotential $W$, it is not obvious that the resulting characteristic projections are MPPs.
However, one can show that the characteristics from $\widehat{{\cal H}}(\vx, \vp) = 0$ trace the same path in $\vx, \vp$ and differ from the characteristics of $\mathcal{H}(\vx, \vp) = 0$ by a time scale.
We leave the details to Appendix \ref{sec:app_pathequiv}.
The problem with using $\widehat{\mathcal{H}}$ is that it is not a convex function of $\vp$, which is a essential property for the GMAM and OUM numerical algorithms.
In the context of stochastic hybrid systems (such as the stochastic ML model), the WKB method by itself does not uniquely define a Hamiltonian, but this can be resolved using a large deviation principle as shown in Appendix \ref{sec:app_path}.
Now that we have identified how to formulate the Hamiltonian using the WKB approximation, it can be used as a stand-alone method.

Given the solution to \eqref{eq:66}, the conditional distribution $r_{0}$ is given by
\begin{equation}
  \label{eq:67}
  r_{0}(n | \vx) = \binom{N}{n}\Lambda(\vx)^{n}(1-\Lambda(\vx))^{N-n},
\end{equation}
where
\begin{equation}
  \label{eq:97}
 \Lambda(\vx) = \frac{A(\vx)}{1+A(\vx)},\quad A(\vx) = a_{\Na}(v) - \frac{\hgamna}{N}(\pv(\vx) g(\vx) + h(\vx, \pw(\vx))).
\end{equation}

\subsubsection{Gaussian approximation near a stable fixed point}
\label{sec:linear-theory}
Near a stable fixed point, a weak noise stochastic process behaves like a continuous Markov process.
Approximating a stochastic process by a simpler continuous Markov process is called a diffusion approximation.
Two widely used methods for obtaining these approximations are the system-size expansion \cite{kampen07a} and a stochastic quasi-steady-state reduction \cite{keener11a}.
It is easy to verify by applying the WKB method to the Fokker--Plank equation that the Hamiltonian for continuous Markov process must be a quadratic function of $\vp$.
Unfortunately, the diffusion approximation breaks down for metastable events like spontaneous excitability.
However, the approximation plays an important role in numerical algorithms (see Appendix \ref{sec:app_oum}).
Moreover, since $\vp = 0$ at fixed points, the linear behavior of Hamilton's equations \eqref{eq:41} near a fixed point is identical to the diffusion approximation.

At the stable fixed point $\vxa$, the stationary probability density function is sharply peaked and is approximately Gaussian.
To see this, Taylor expand the quasipotential near $\vxa$ (with $W(\vxa) = 0$) to get
\begin{equation}
  \label{eq:83}
  W(\vx)\sim   \frac{1}{2}(\vx - \vxa)^{T}Z(\vxa)(\vx - \vxa ) + \cdots,
\end{equation}
where $Z$ is the Hessian matrix,
\begin{equation}
\label{eq:58}
  Z_{s, s'} \equiv \frac{\partial^{2}W}{\partial s \partial s'}.
\end{equation}
On characteristics $Z$ satisfies the Ricatti equation \cite{ludwig75a,maier97a}
\begin{equation}
\label{eq:38}
  \frac{d}{dt}Z = -ZDZ - ZC - C^{T}Z - G,
\end{equation}
where
\begin{equation}
\label{eq:42}
  D_{s, s'}(\vx, \vp) = \frac{\partial^{2}\mathcal{H}}{\partial p_{s}\partial p_{s'}},\quad
  C_{s, s'}(\vx, \vp) = \frac{\partial^{2}\mathcal{H}}{\partial p_{s}\partial s'},\quad
  G_{s, s'}(\vx, \vp) = \frac{\partial^{2}\mathcal{H}}{\partial s\partial s'}.
\end{equation}
Characteristics converge to stable (unstable) fixed points in the limit $t\to-\infty$ ($t\to\infty$).
Therefore, at these points we have that $Z'  = 0$.  
Furthermore, we know that $\vp = 0$ at fixed points and that $\mathcal{H}(\vx, 0) = 0$.
Hence, at the fixed point \eqref{eq:38} becomes
\begin{equation}
  \label{eq:88}
   ZDZ + ZC + C^{T}Z  = 0, 
\end{equation}
which is called the {\em algebraic Ricatti equation}.
The solution to this equation yields a Gaussian approximation of the local stationary probability density, which suggests that the process behaves like a continuous Markov process near the stable fixed point.

To see the connection between the stochastic ML model and a diffusion approximation we first expand $\mathcal{H}$ near the stable fixed point.
Because $\vp = 0$ at fixed points, \eqref{eq:88} can be simplified by expanding $\mathcal{H}$ around $\vp = 0$ and $\vx = \vxa$.
Expanding to second order in $\vp$ and $\vx-\vxa$ yields
\begin{equation}
  \label{eq:90}
  \widetilde{\mathcal{H}}(\vx, \vp) \equiv  \sum_{s = v, w}p_{s}\sum_{s' = v, w}(s' - s'_{A})C_{s,s'}(\vxa , 0) 
+ \frac{1}{2}\sum_{s = v, w}p_{s}^{2}D_{s, s}(\vxa , 0).
\end{equation}
Note that many terms vanish at the fixed point, namely, $\mathcal{H}(\vxa, 0) = 0$, $\nabla_{\vx}\mathcal{H}(\vxa, 0) = \nabla_{\vp}\mathcal{H}(\vxa, 0) = 0$, and $G_{s, s'}(\vxa, 0) = 0$.
The above is consistent with the Hamiltonian of a continuous Markov process corresponding to the Fokker--Plank equation,
\begin{multline}
  \label{eq:62}
  \pd{}{t}\tilde{\rmp}(\vx, t) =  -\sum_{s = v, w}\pd{}{s}\left[\sum_{s' = v, w}(s' - s'_{A})C_{s, s'}(\vxa , 0)\tilde{\rmp}(\vx, t)\right] \\
+ \frac{\epsilon}{2}\sum_{s = v, w}\pdd{}{s}\left[D_{s, s}(\vxa , 0)\tilde{\rmp}(\vx, t) \right].
\end{multline}
It is a useful exercise to show that a WKB expansion of \eqref{eq:62} results in the Hamiltonian \eqref{eq:90}.

\section{Results}
\subsection{Type I excitability}
\label{sec:sp1}
\subsubsection{Spontaneous action potentials}
Consider the parameter regime where the deterministic system has three fixed points: one stable, one saddle, and one unstable.
(Parameter values are listed in Appendix \ref{sec:type-i}.)
A phase plane diagram of the deterministic dynamics is shown in Fig.~\ref{fig:type1a_pp1}(a).
\begin{figure}[tbp]
  \centering
  \includegraphics[width=10cm]{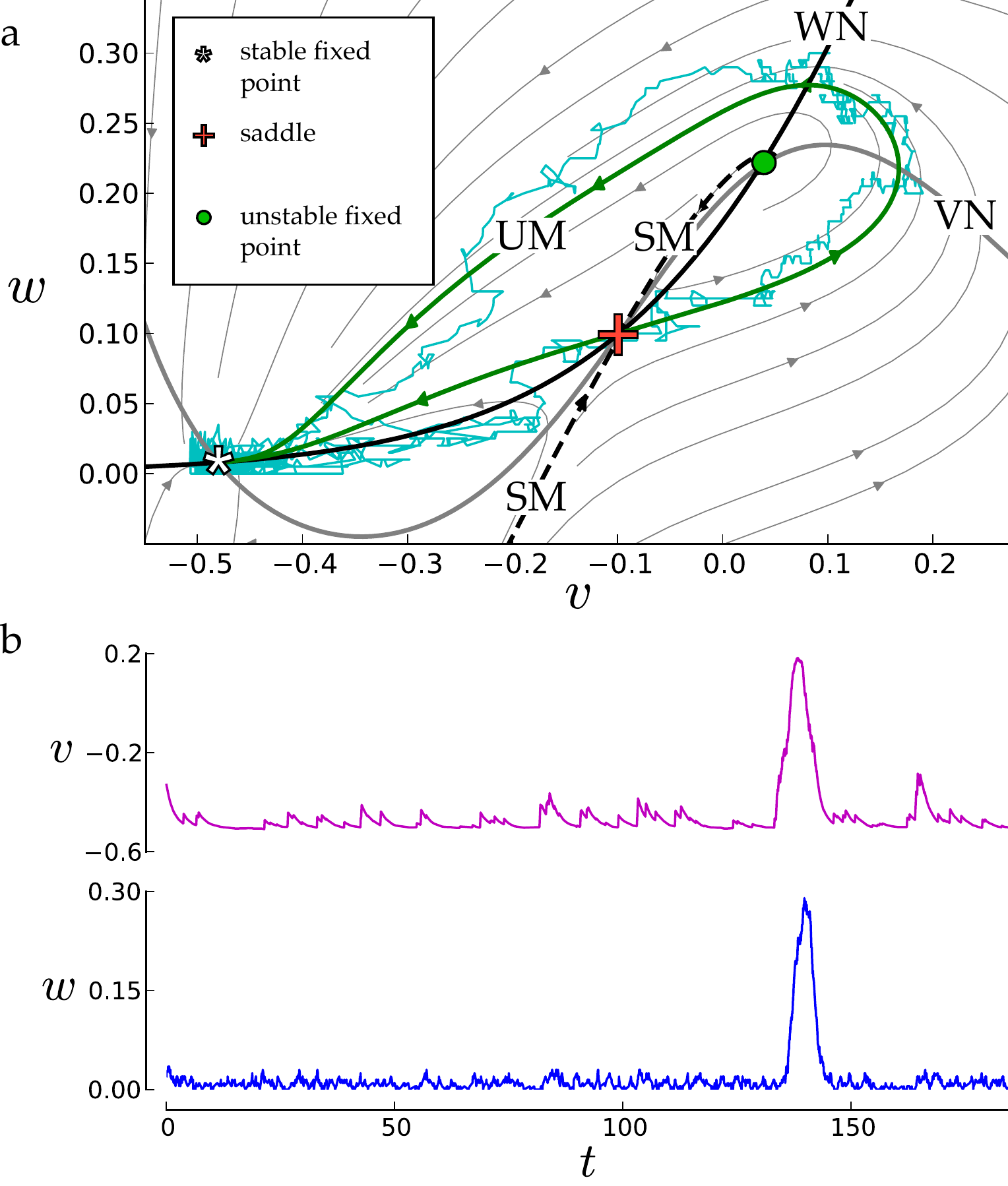}
  \caption{(a) Deterministic phase plane for type I excitability. Streamlines of the deterministic vector field are shown as thin grey curves.  A representative stochastic trajectory of an excitable event is shown in blue. Also labeled in the figure are the $v$-nullcline (VN),  $w$-nullcline (WN), stable manifold of the saddle (SM), and unstable manifold of the saddle (UM). (b) Representative time dependent stochastic trajectory of an excitable event.}
  \label{fig:type1a_pp1}
\end{figure}
The stable manifold of the saddle defines a threshold for excitation.
A deterministic trajectory starting to the left of the threshold quickly converges to the stable fixed point.
On the other hand, when starting to the right of the threshold, the trajectory exhibits a transient increase in voltage as it travels around the unstable fixed point before reaching the stable fixed point.
Hence, a noise induced action potential can be broken into two phases: a slower initiation phase and a faster transient spike in voltage followed by a return to the stable fixed point.
The initiation phase is a fluctuation-induced spontaneous transition from the stable fixed point to the threshold.
Once the threshold is reached, the return to the stable fixed point is dominated by the deterministic forces rather than ion channel fluctuations.
The most likely path taken during the return phase is along one of the two branches of the unstable manifold (see the green curve Fig.~\ref{fig:type1a_pp1}(a)).
The right branch leads to an excitable event and the left branch leads directly back to the stable fixed point.
Fig.~\ref{fig:type1a_pp1}(b) shows a representative stochastic trajectory obtained by simulation (see Appendix \ref{sec:monte-carlo-simul}).

\begin{figure}[tbp]
  \centering
  \includegraphics[width=10cm]{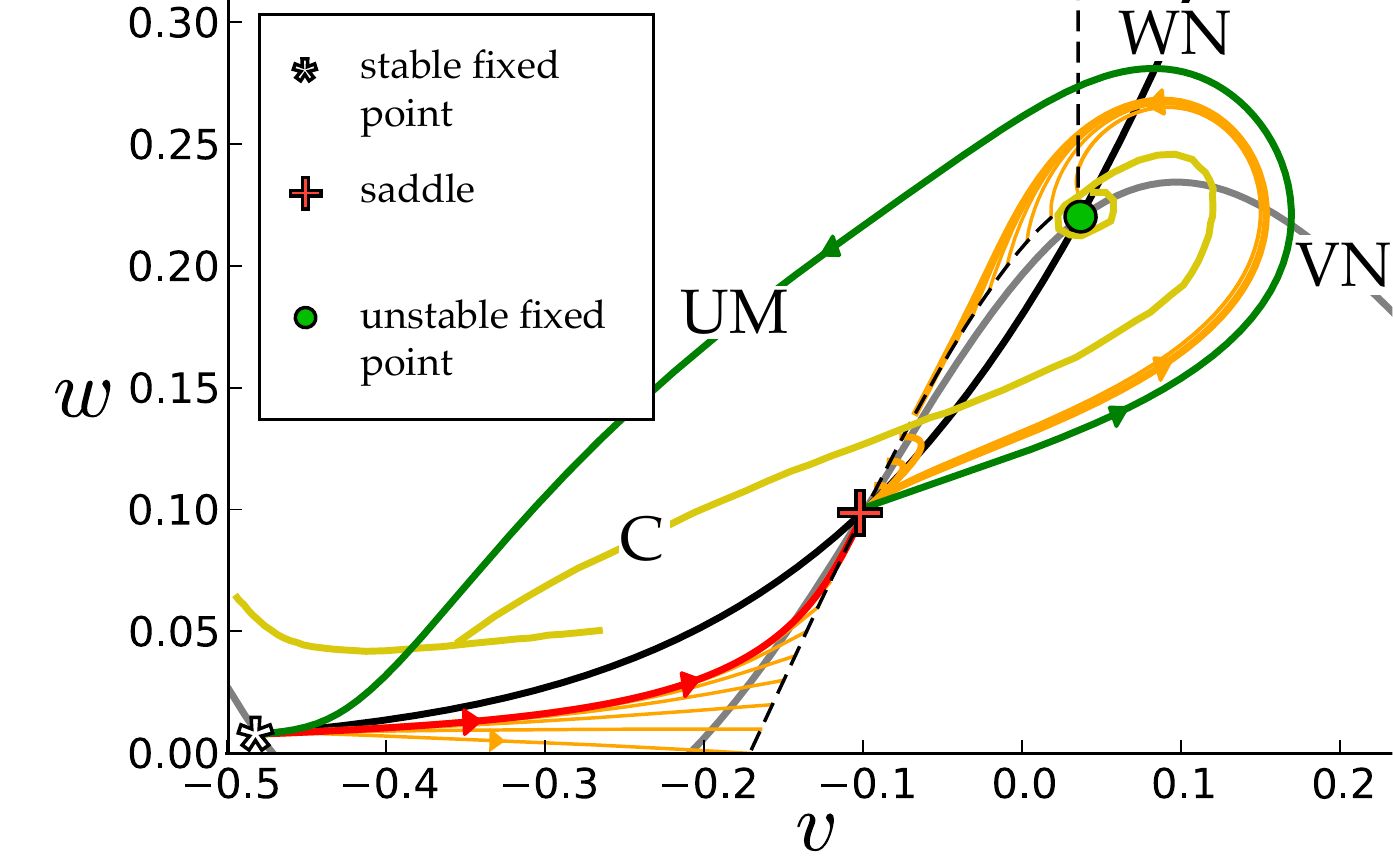}
  \caption{Type I excitability, showing MPPs that start at the stable fixed point and end at the threshold (dashed curve). This red line shows the MPEP that connects the stable fixed point to the saddle.  Multiple orange curves show MPPs that reach other points on the threshold. Also shown are the caustics (C) where characteristic projections (not shown) collide.}
  \label{fig:type1a_pp2}
\end{figure}
\begin{figure}[tbp]
  \centering
  \includegraphics[width=12cm]{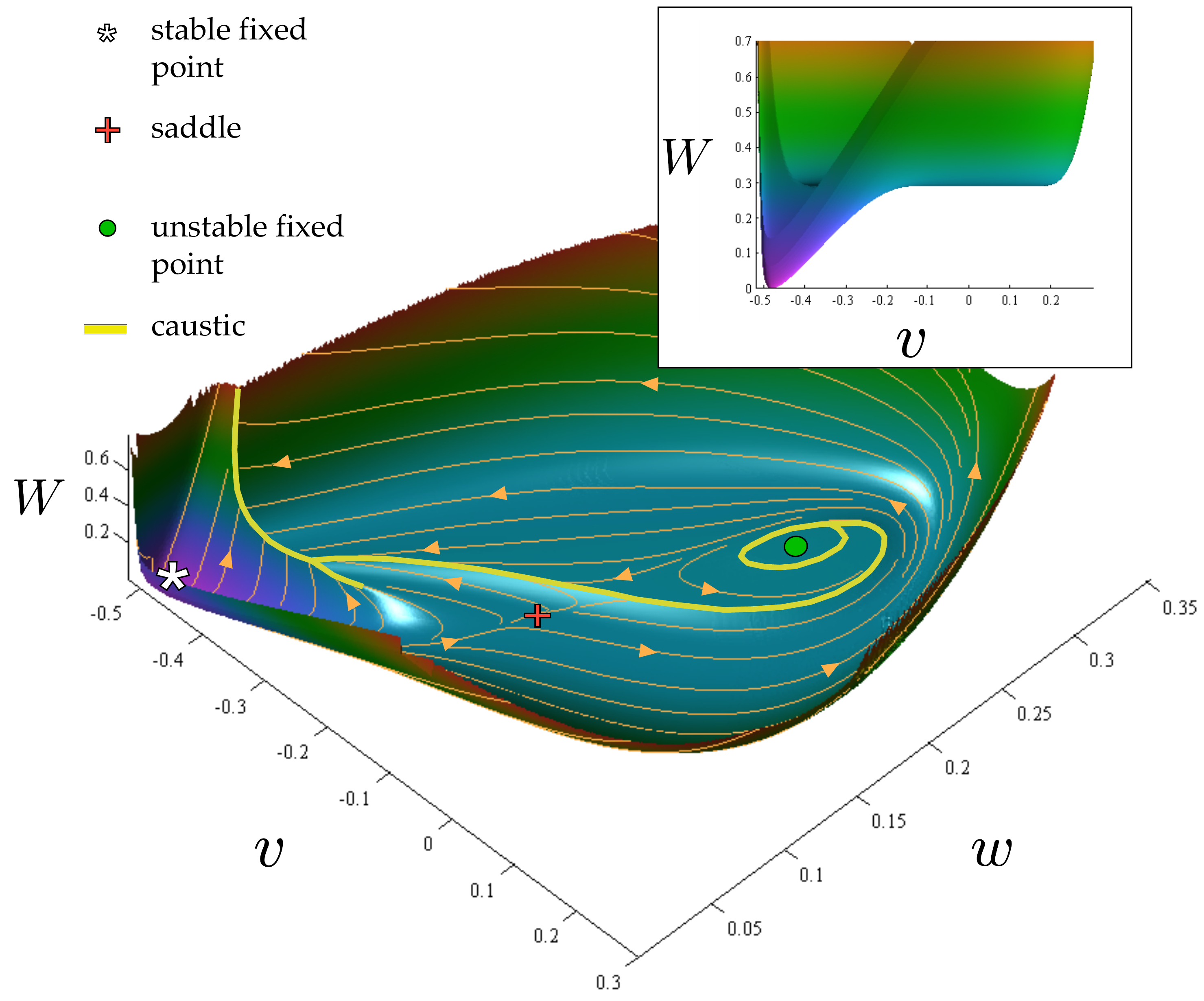}
  \caption{Type I excitability, showing the quasipotential $W(v, w)$ computed using the OUM on a $500\times500$ grid.  Yellow curve shows the caustic.  Orange streamlines show the behavior of MPPs.}
  \label{fig:type1a_qp}
\end{figure}
Fig.~\ref{fig:type1a_pp2} shows a several MPPs (orange curves, computed using the GMAM) that start at the stable fixed point.
MPPs can only directly lead to the threshold (dashed line) below the saddle, where they cross from left to right.
Above the saddle, MPPs first travel through the saddle and then cross the threshold from right to left.
Hence, MPPs that cross the threshold above the saddle are most likely to be returning action potential trajectories.
That is, action potential trajectories at the end of the excitation phase that are returning to the stable resting potential.
Spontaneous action potentials are most likely initiated below the saddle.

Above the stable fixed point many characteristic projections overlap.
Uniqueness of the solution is recovered using the minimum action principle (see Section \ref{sec:mpp}).
At any given point $\vx$ through which two or more characteristics cross, the value of $W(\vx)$ is given by the characteristic that has the smallest action.
A caustic is a curve along which each point is the terminus of two ore more MPPs and the gradient of the quasipotential is discontinuous.
Fig.~\ref{fig:type1a_pp2} shows two branches (yellow curves) of the caustic.
To the left of the threshold, transient excursions from the attractor reach the caustic from below, while returning action potentials reach the caustic from above.
The upper branch of the caustic crosses the separatrix above the saddle, separating excitation-phase MPPs into those that travel around the unstable fixed point and those that do not.

The quasipotential, computed using the OUM, is shown in Fig.~\ref{fig:type1a_qp}.
As there are two phases to an excitable event, there are two regions of the quasipotential, separated by the threshold.
Around the stable fixed point is a potential well from which the trajectory must escape during the initiation phase.
Around the unstable fixed point, to the right of the threshold, the quasipotential forms a horseshoe-canyon-like shape (see Fig.~\ref{fig:type1a_qp_2}), the bottom of which lies flat along the unstable manifold (see green curve in Fig.~\ref{fig:type1a_pp2}).
\begin{figure}[htbp]
  \centering
  \includegraphics[width=10cm]{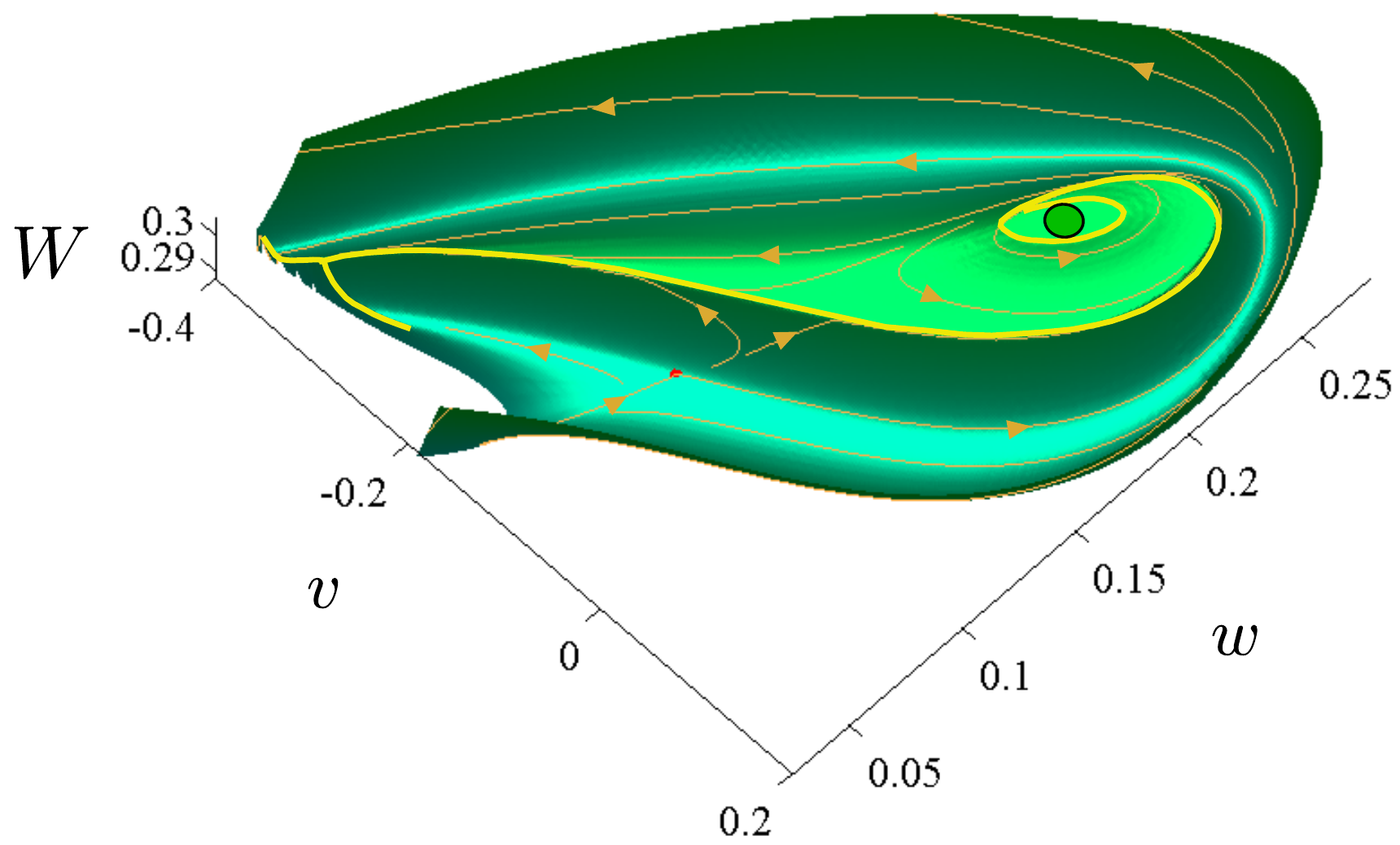}
  \caption{A close up of the quasipotential from Fig.~\ref{fig:type1a_qp} around the unstable fixed point.}
  \label{fig:type1a_qp_2}
\end{figure}
Hence, the most probable MPP for the excitation phase is the deterministic unstable manifold.

While the unstable manifold wraps around the unstable fixed point only once, a stochastic trajectory can rotate around the unstable fixed point during the excitation phase, prolonging the action potential.
The situation can be more pronounced if there is a stable limit cycle surrounding the unstable fixed point as shown in the next section.

\subsubsection{Spontaneous bursting}
By slightly changing parameters, a stable limit cycle can emerge around the unstable fixed point corresponding to the excited state.
(Parameter values are listed in Appendix \ref{sec:type-i-with}.)
The phase plane is shown in Fig.~\ref{fig:type1b_pp1}(a).
\begin{figure}[tbp]
  \centering
  \includegraphics[width=10cm]{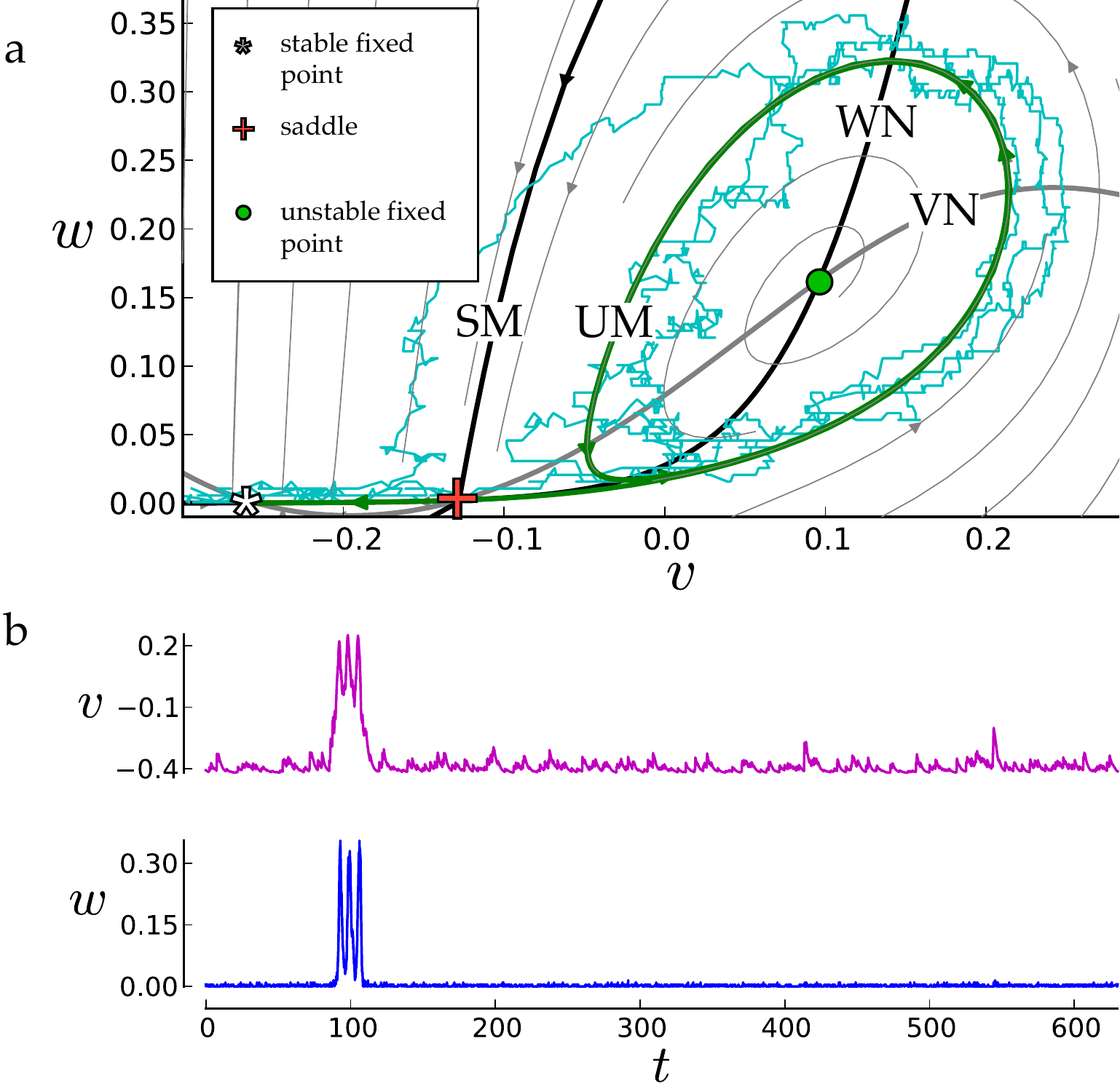}
  \caption{Type I excitability with bursting, showing (a) the deterministic phase plane. Streamlines of the deterministic vector field are shown as thin grey curves.  A representative stochastic trajectory of an excitable event is shown in blue.  (b) Representative time dependent stochastic trajectory of an excitable event.}
  \label{fig:type1b_pp1}
\end{figure}
The unstable manifold (green curve) converges to a stable limit cycle surrounding the unstable fixed point.
Also shown is a representative stochastic trajectory, Fig.~\ref{fig:type1b_pp1}(b).

With a stable limit cycle surrounding the unstable fixed point, the excited state is no longer transient, and escape back to the resting state is also a spontaneous event.
Hence, a single excitable event is characterized by two spontaneous events occurring in sequence: spontaneous initiation and escape from the excited state back to the resting state.
If the average duration of an excited state is much less than average time for spontaneous initiation, a single excitable event can be described as spontaneous bursting.
In contrast to deterministic bursting where the number of bursts is fixed, the duration of the burst is also random.

Fig.~\ref{fig:type1b_pp2} shows MPPs from the stable fixed point to the saddle (red) and from the stable limit cycle to the saddle (blue).
\begin{figure}[tbp]
  \centering
  \includegraphics[width=10cm]{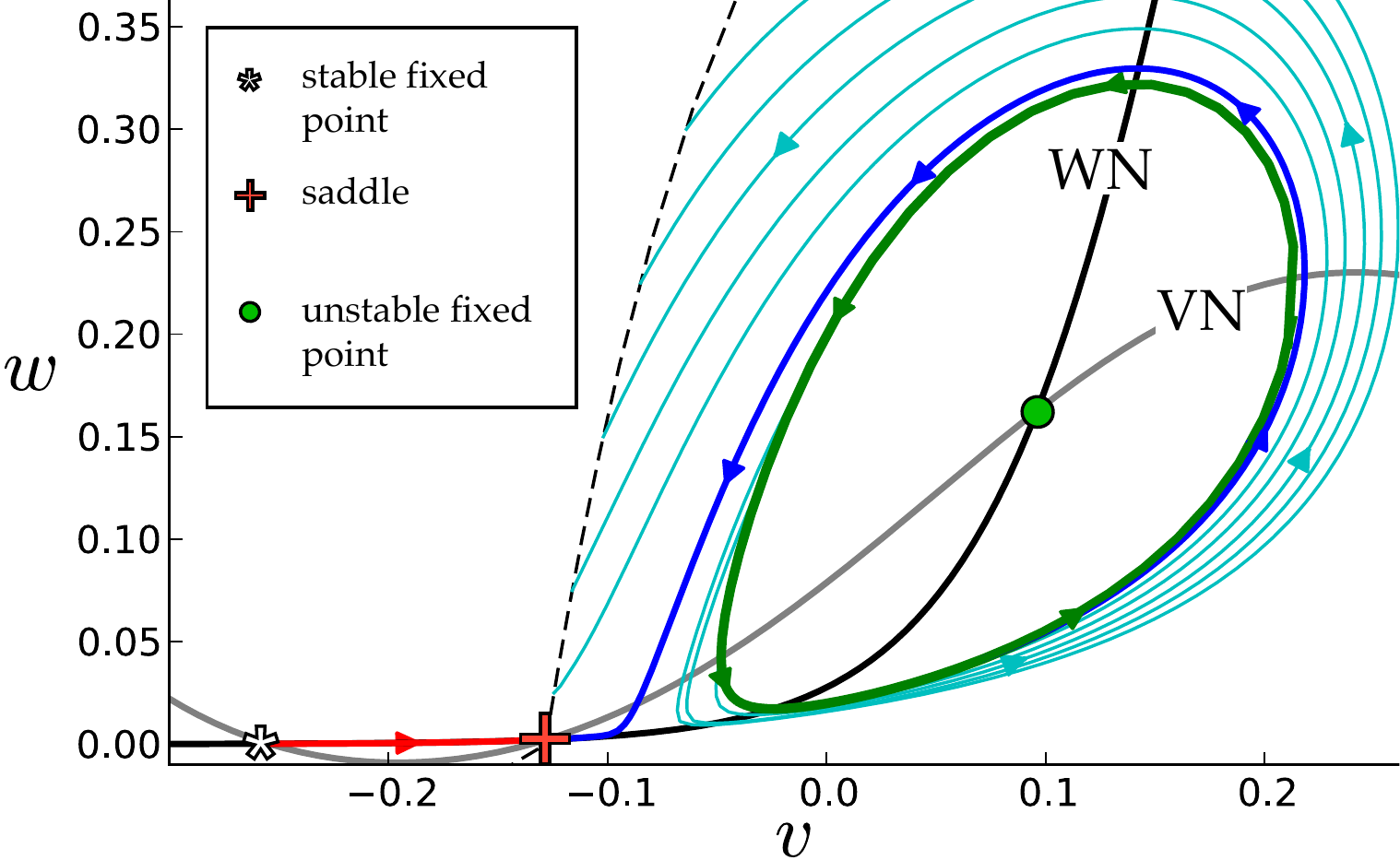}
  \caption{Type I excitability with bursting, showing MPPs that start at the stable fixed point and reach the threshold (dashed curve). This red line shows the MPP that connects the stable fixed point to the saddle. Note that MPPs from the stable fixed point do not cross the saparatrix above the saddle. The blue line shows the MPP that connects the stable limit cycle (green) with the saddle (S).  Light blue curves show MPPs from the stable limit cycle that cross the saparatrix above the saddle.}
  \label{fig:type1b_pp2}
\end{figure}
\begin{figure}[tbp]
  \centering
  \includegraphics[width=12cm]{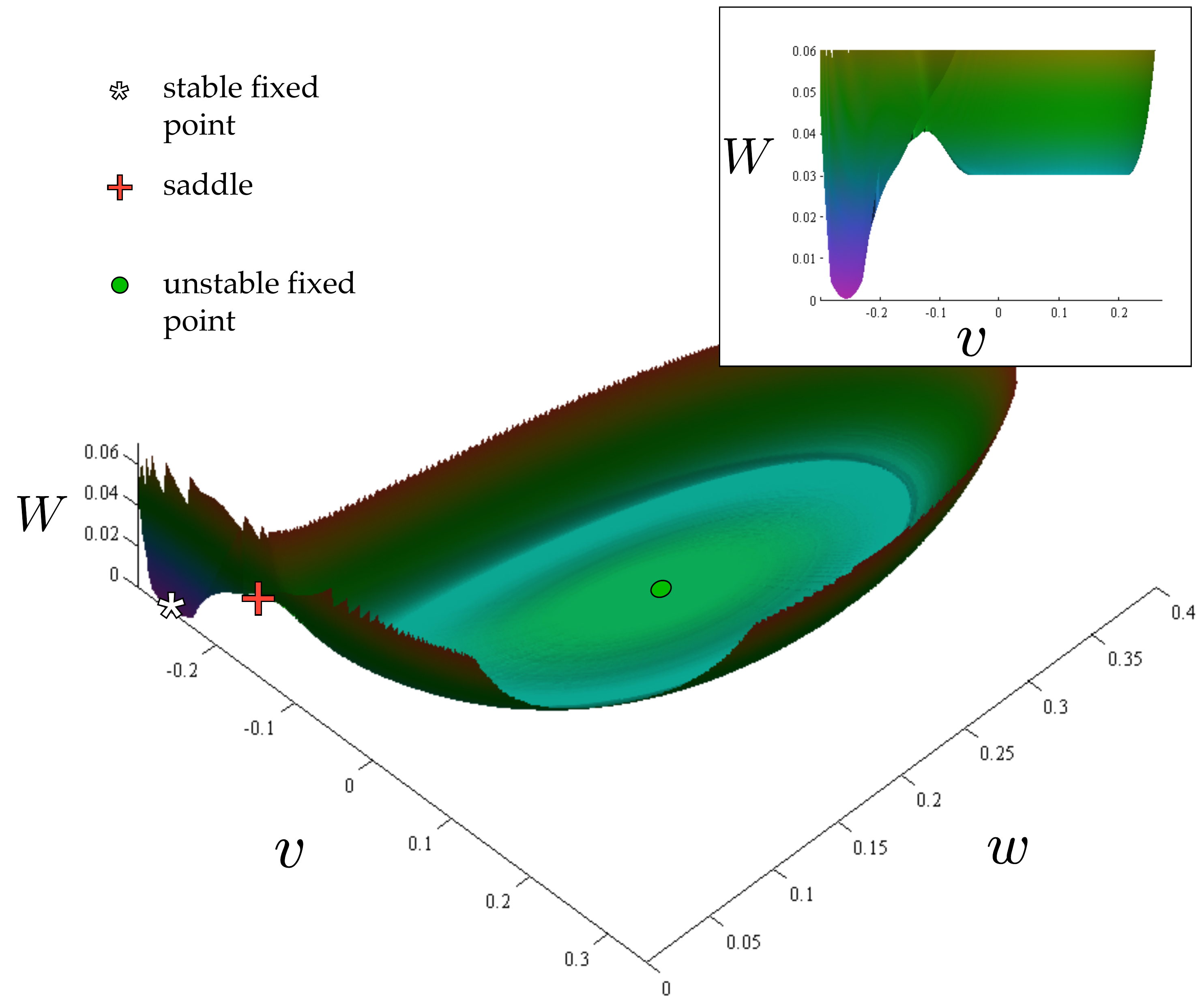}
  \caption{Type I excitability with bursting, showing the quasipotential $W(v, w)$ computed using the OUM on a $500\times500$ grid.}
  \label{fig:type1b_qp}
\end{figure}
In this case, we have that $\mu > 1$, where $\mu$ is the eigenvalue ratio at the saddle.
Hence the MPEP is not tangent to the threshold as in the previous example.
Like the previous example, MPPs that start at the stable fixed points reach the threshold below the saddle only, while MPPs that start from the limit cycle (light blue) cross the threshold above the saddle only.
Notice that there is an effective reflecting barrier near the saddle along $w = 0$ so that all of the MPPs that reach the threshold from the stable fixed point are very close.
We can expect that the exit behavior of the initiation event is nearly one dimensional as $w$ is approximately fixed along the red curve.

The quasipotential is shown in Fig.~\ref{fig:type1b_qp}.
The front view (Fig.~\ref{fig:type1b_qp}(inset)) shows that the quasipotential has the profile of a double well potential, with a local minimum of the left potential well at the stable fixed point and a local maximum at the saddle.
However, the profile appears flat at the bottom of the right potential well, corresponding to the excited state, because of the stable limit cycle.
Notice that the left well is very thin (with respect to $w$), confirming that $w$ is approximately constant during the initiation phase.

The right well takes a shape similar to the bottom of a wine bottle (see Fig.~\ref{fig:type1b_qp_2}).
\begin{figure}[htbp]
  \centering
  \includegraphics[width=10cm]{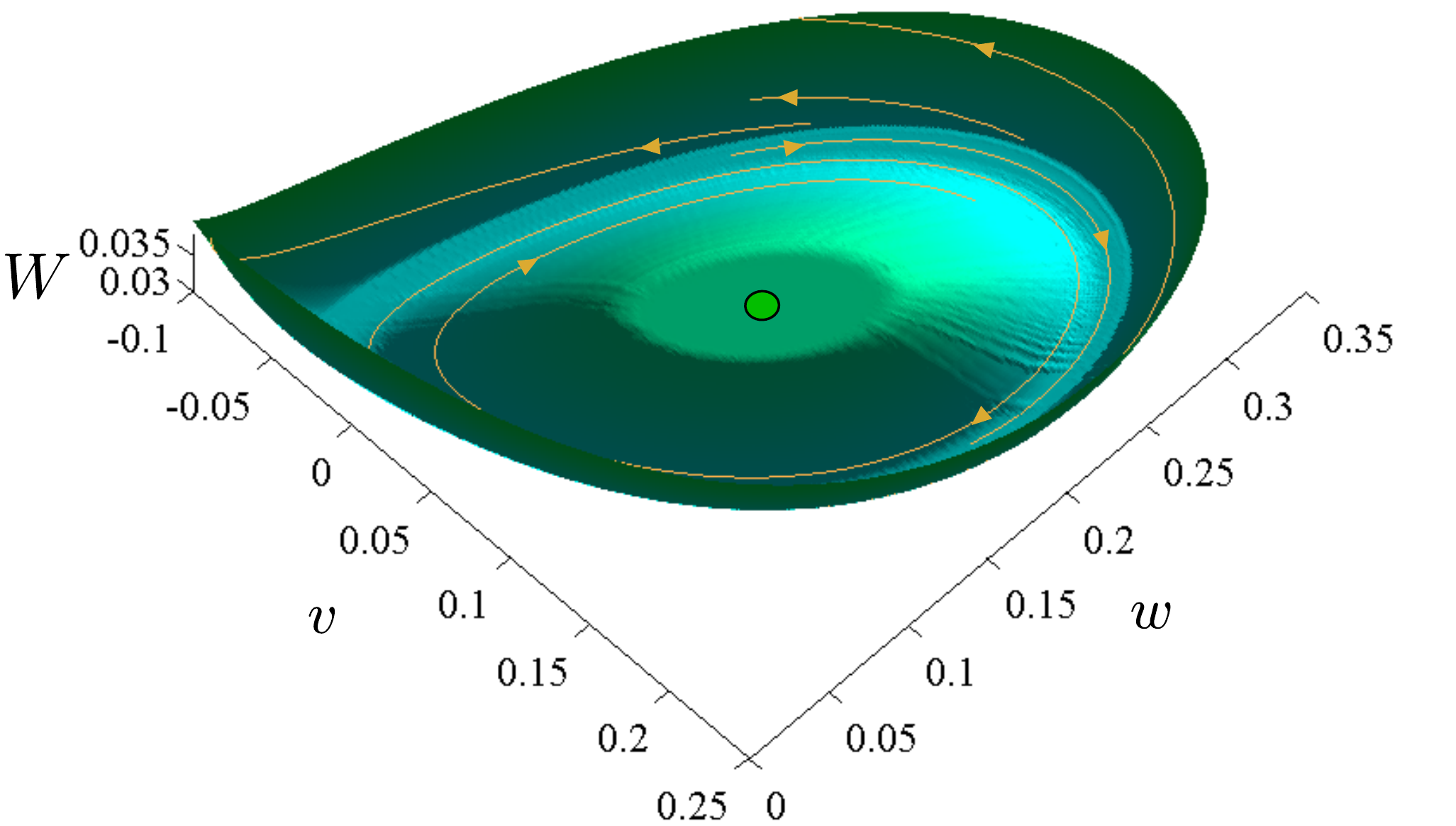}
  \caption{A close up of the quasipotential from Fig.~\ref{fig:type1b_qp} around the unstable fixed point.}
  \label{fig:type1b_qp_2}
\end{figure}
Along the deterministic limit cycle $\vx_{c}(t)$, the quasipotential is constant with $W(\vx_{c}(t)) \equiv W_{c}$.
For $\vx$ in a neighborhood of $\vx_{c}$, the quasipotential is increasing so that $W_{c} \leq W(\vx)$.
On the interior of the limit cycle, $W$ has a local maximum at the unstable fixed point.

\subsection{Type II excitability}
\label{sec:sp2}
Excitability is also possible with a single stable fixed point and no unstable fixed points.
(Parameter values are listed in Appendix \ref{sec:type-ii}.)
Fig.~\ref{fig:type2a_pp1} shows the deterministic phase plane for a type II excitable system.
Also shown is a representative stochastic trajectory that starts at the stable fixed point and later undergoes an excitable event.
Similar to type I spontaneous excitability, the stochastic trajectory spends a long period of time near the stable fixed point until a it undergoes an initiation event.
After initiation, it follows close to a deterministic trajectory until it returns to the stable fixed point.

In contrast to the previous examples, it is difficult to define a threshold without an unstable saddle point which defines a separatrix.
\begin{figure}[tbp]
  \centering
  \includegraphics[width=10cm]{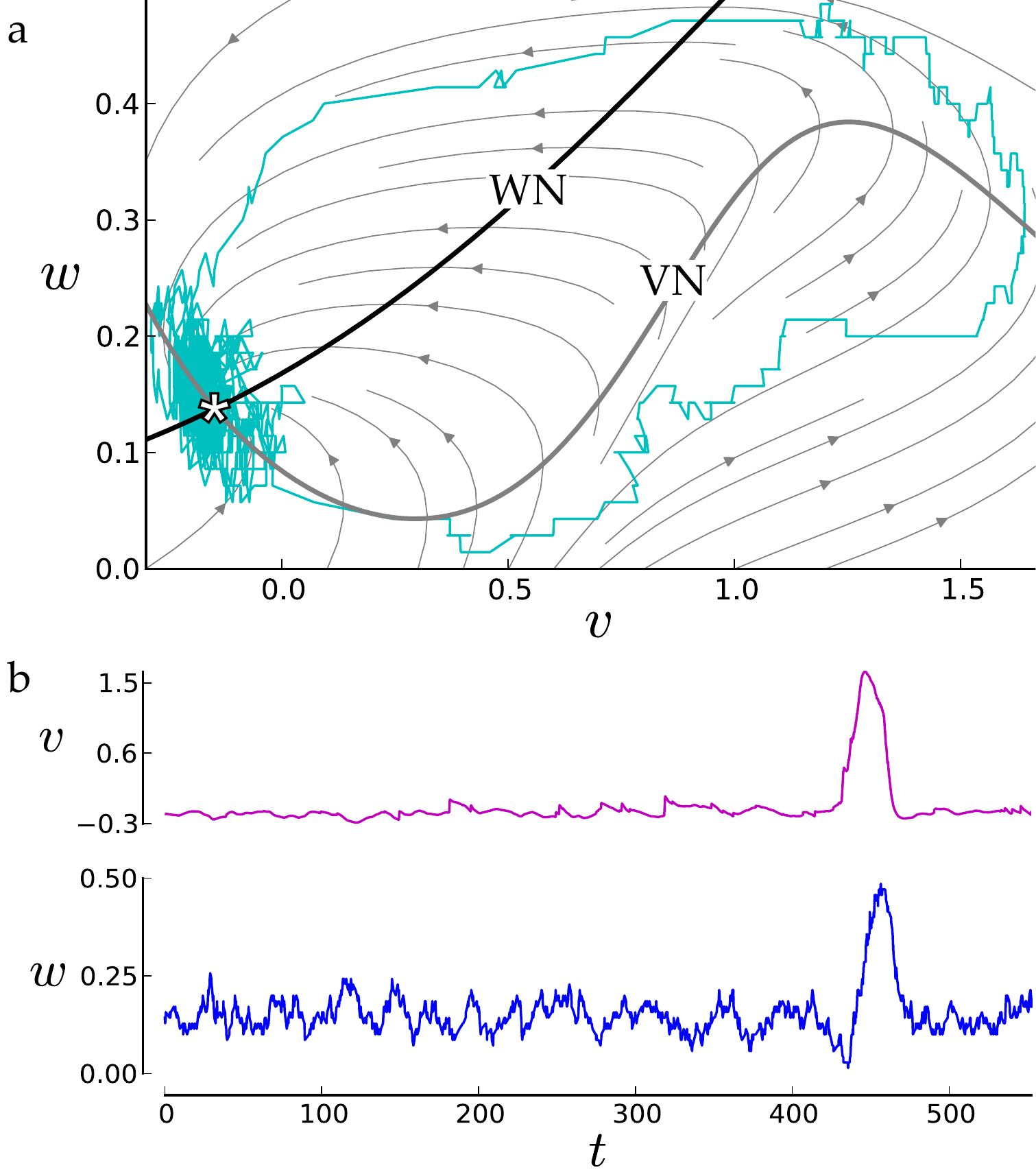}
  \caption{(a) Deterministic phase plane for type II excitability. Streamlines of the deterministic vector field are shown as thin grey curves.  A representative stochastic trajectory of an excitable event is shown in blue.  (b) Representative time dependent stochastic trajectory of an excitable event.}
  \label{fig:type2a_pp1}
\end{figure}
For the deterministic limit, a threshold can be defined in certain limits \cite{khovanov13a}, but there is no general structural definition and often ad-hoc excitation thresholds are defined.
However, a stochastic analysis of MPPs shows that the excitable system with one fixed point shares many features in common with the system with three fixed points.

MPPs that lead to spontaneous action potentials are shown in Fig.~\ref{fig:type2a_pp2}.
\begin{figure}[tbp]
  \centering
  \includegraphics[width = 10cm]{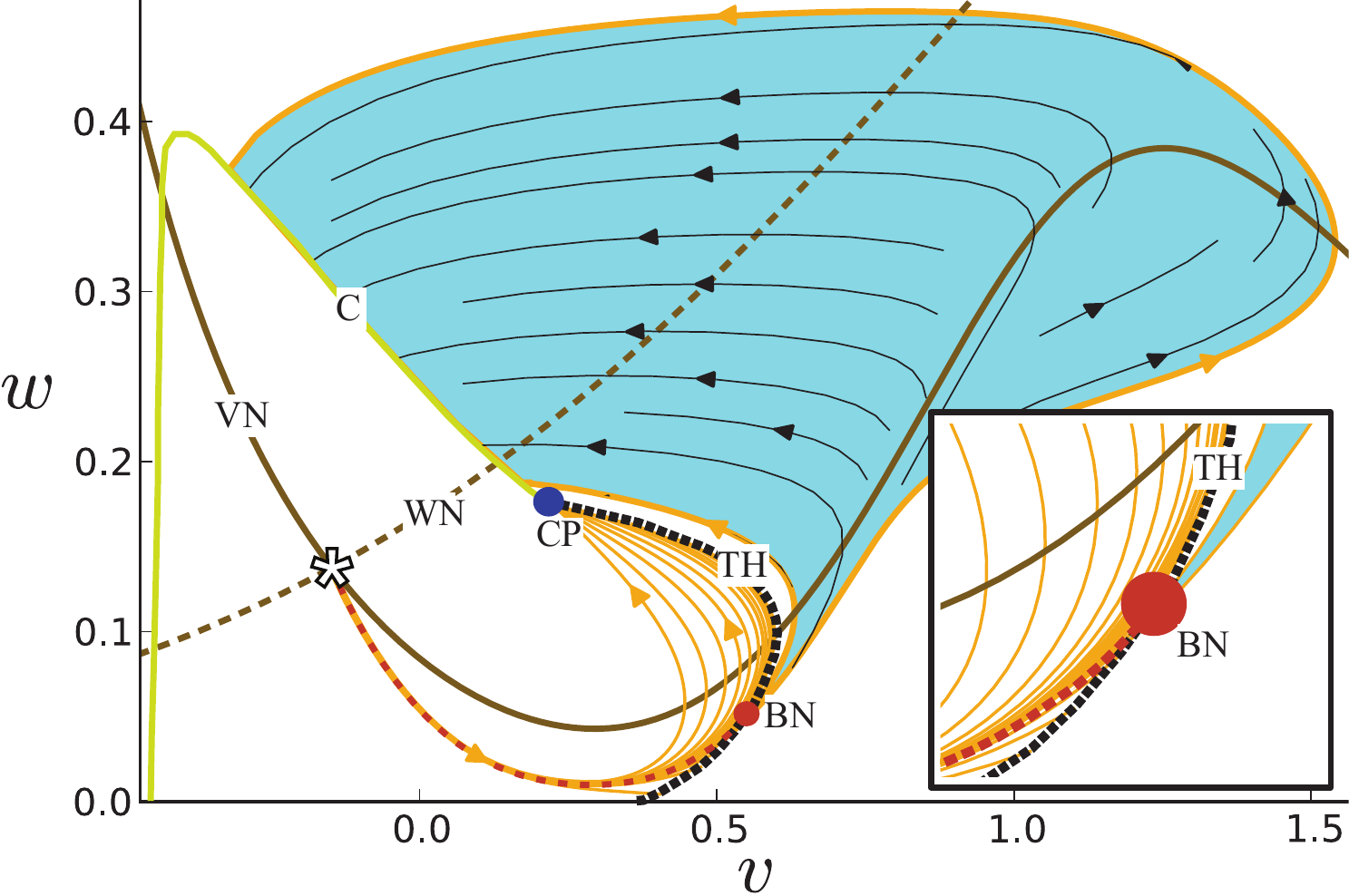}
  \caption{Type II excitability, showing MPPs (orange curves) during initiation. The dashed black curve (TH) shows a level curve of $W$ that reaches the caustic formation point (CP). The blue region contains the most observable action potential trajectories during the excitation phase.  The black stream lines show the deterministic dynamics. During the initiation phase, MPPs follow a single path (dashed red line) and cross through a bottleneck (BN) before reaching the excitation phase (blue region).}
  \label{fig:type2a_pp2}
\end{figure}
The blue region, which corresponds to a region where $W$ is relatively flat, contains the most likely observable action potential trajectories during the excitation phase.
A threshold (dashed black line TH in Fig.~\ref{fig:type2a_pp2}) is defined as the level curve of $W$ passing through the caustic formation point.
All of the MLTs that eventually cover the blue region start out very close together during the initiation phase and cross the threshold very close to a single point called the bottleneck (see Fig.~\ref{fig:type2a_pp2} inset).
In other words, the most observable spontaneous action potential trajectories follow a single path (red dashed curve) during initiation.
We expect to see this in a type I system because of the saddle, and it is interesting that this behavior is preserved in a type II system.
The bottleneck acts like a saddle even though the Type II system has no saddle.

If MPPs asymptotically approach deterministic trajectories, it must be in the limit $t\to\infty$.
In a type I excitable system, this happens at the end of the initiation phase as the MPEP approaches the stable manifold of the saddle.
On the other hand, for a type II excitable system there is no saddle, and asymptotic convergence to a deterministic trajectory does not occur until the end of the return phase when it reaches the stable fixed point.
However, as is shown in Fig.~\ref{fig:type2a_pp2}, after MPPs leave the potential well region, they become very close to deterministic trajectories.
This can be quantified by observing that $\norm{\vp}$ is small (but not zero) after leaving the potential well region.
Unlike the type I example, a stochastic trajectory, once it reaches the boundary of the potential well, is not equally likely to continue to become an action potential or return directly to the stable fixed point.
However, since $\norm{\vp}$ is small, the probability of returning directly to the stable fixed point is slightly higher than to generate an action potential.
Hence, the mean time to exit the potential well region can be viewed as a lower bound on the mean time to initiate an action potential, the former being less than half the value of the latter.

The quasipotential, computed numerically, is shown in Fig.~\ref{fig:type2a_qp}.
\begin{figure}[tbp]
  \centering
  \includegraphics[width=12cm]{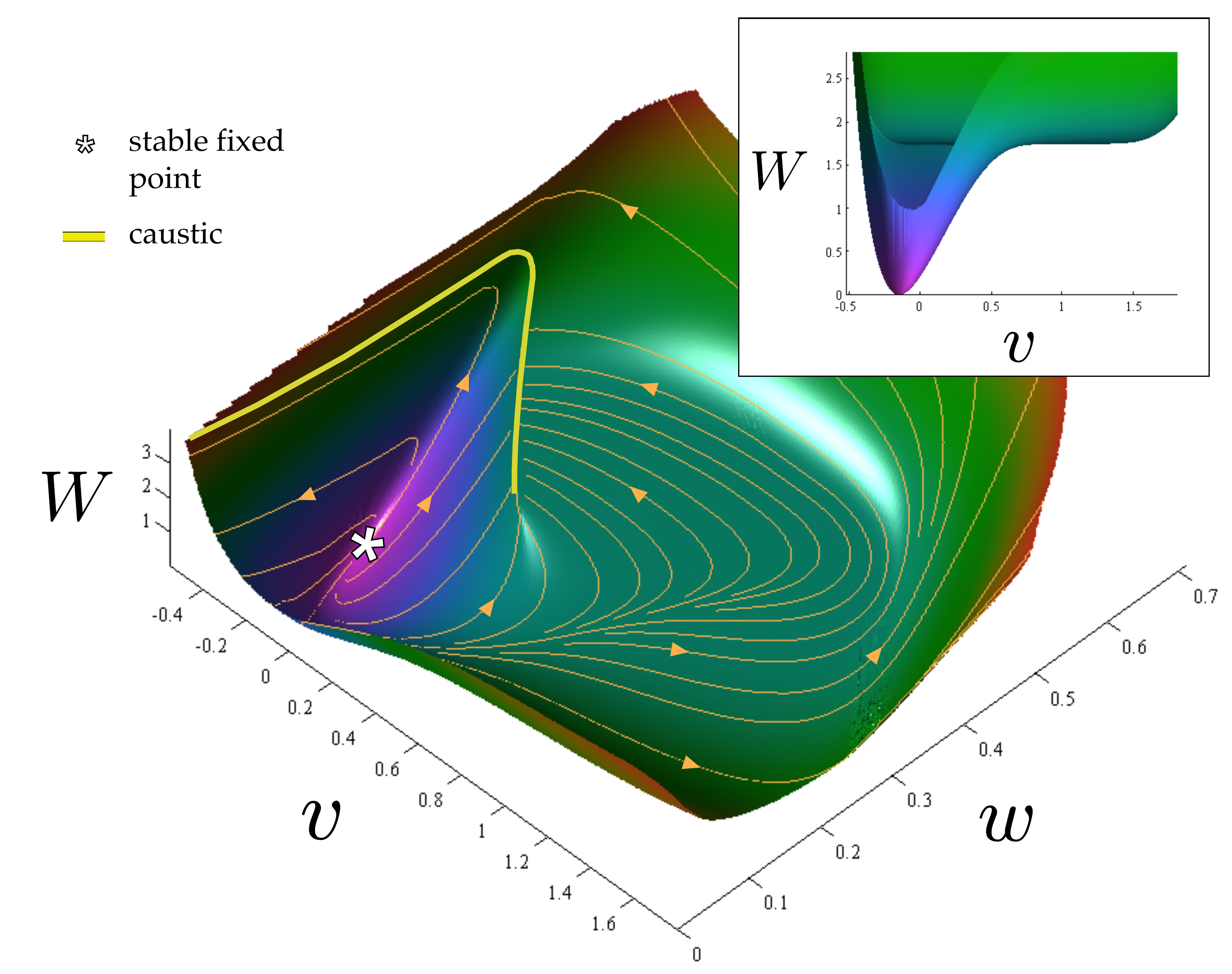}
  \caption{Type II excitability, showing the quasipotential $W(v, w)$ computed using the OUM on a $500\times500$ grid. Yellow curve shows the caustic. Orange streamlines show the behavior of MPPs.}
  \label{fig:type2a_qp}
\end{figure}
The orange streamlines represent segments of MPPs that start at the stable fixed point.
As shown in Fig.~\ref{fig:type2a_qp}, there is a caustic formation point where MPPs begin to overlap.
To the right of this point, there is a region where MPPs get very close to deterministic trajectories during the excitation phase.
That is, during the excitation phase, MPPs are very similar to deterministic action potentials.
This matches the intuition that noise must push the system away from the stable fixed point into the region where deterministic trajectories undergo an excitable event.
In the deterministic system \eqref{eq:28} the fixed point is stable.
In the higher dimensional Hamiltonian dynamical system \eqref{eq:41}, the fixed point is a saddle, with a stable manifold tangent to $\vp=0$ (recall that setting $\vp=0$ recovers the deterministic limit).
Hence, there is a subset of MPPs that form closed trajectories that are heteroclinic connections between the unstable and stable manifold.
As these MPPs return toward the fixed point, they collide with unstable MPPs leading away from the fixed point along a caustic.
The caustic wraps around the potential well region, marking the edge of a flat shelf-like region corresponding to the excited state.

Finally, to check the accuracy of the result, we use Monte Carlo simulations to confirm that the behavior of the process during the initiation phase is consistent with the MPPs shown in Fig.~\ref{fig:type2a_pp2}.
Motivated by Ref.~\cite{dykman96a}, we perform Monte-Carlo simulations (for details about the algorithm, see Appendix \ref{sec:monte-carlo-simul}) to obtain trajectories that start at the stable fixed point and eventually reach the line $v = 0.6$.
Several simulation trajectories are shown in Fig.~\ref{fig:mctraj_type2a} along with the MPP (dashed red curve) that corresponds to the dashed red curve shown in Fig.~\ref{fig:type2a_pp2}.
From the ensamble of these trajectories, we determine the statistics of the position as a function of time preceding arrival at the threshold.
We then set the time at which each trajectory ends (i.e., the time at which they reach $v=0.6$) to $t=0$ and look backward in time in order to observe the behavior of the process during the initiation phase of a spontaneous action potential.

Trajectories are sampled to obtain histograms of the {\em path history} defined as probability density $Q(\vx, t | \vx_{f}, t_{f}; \vxa, t_{0})$, for $t_{0} < t < t_{f}$.
We can express the path history as
\begin{equation}
  \label{eq:23}
  Q(\vx, t) = \frac{\rmp(\vx_{f}, t_{f} | \vx, t)\rmp(\vx, t | \vxa, t_{0})}{\rmp(\vx_{f}, t_{f} | \vxa, t_{0})},
\end{equation}
where $v_{f} = 0.6$, $t_{f} = 0$ and (effectively) $t_{0} = - \infty$ (trajectories take a long time to reach $v_{f}$).
Each pane in Fig.~\ref{fig:mc_type2} represents a discrete approximation from $10^{3}$ simulation trials of $Q$ at a different point in time. 
The most probable path (dashed red curve) coincides with the peak of the histogram as a function of time.
\begin{figure}[tbp]
  \centering
  \includegraphics[width=12cm]{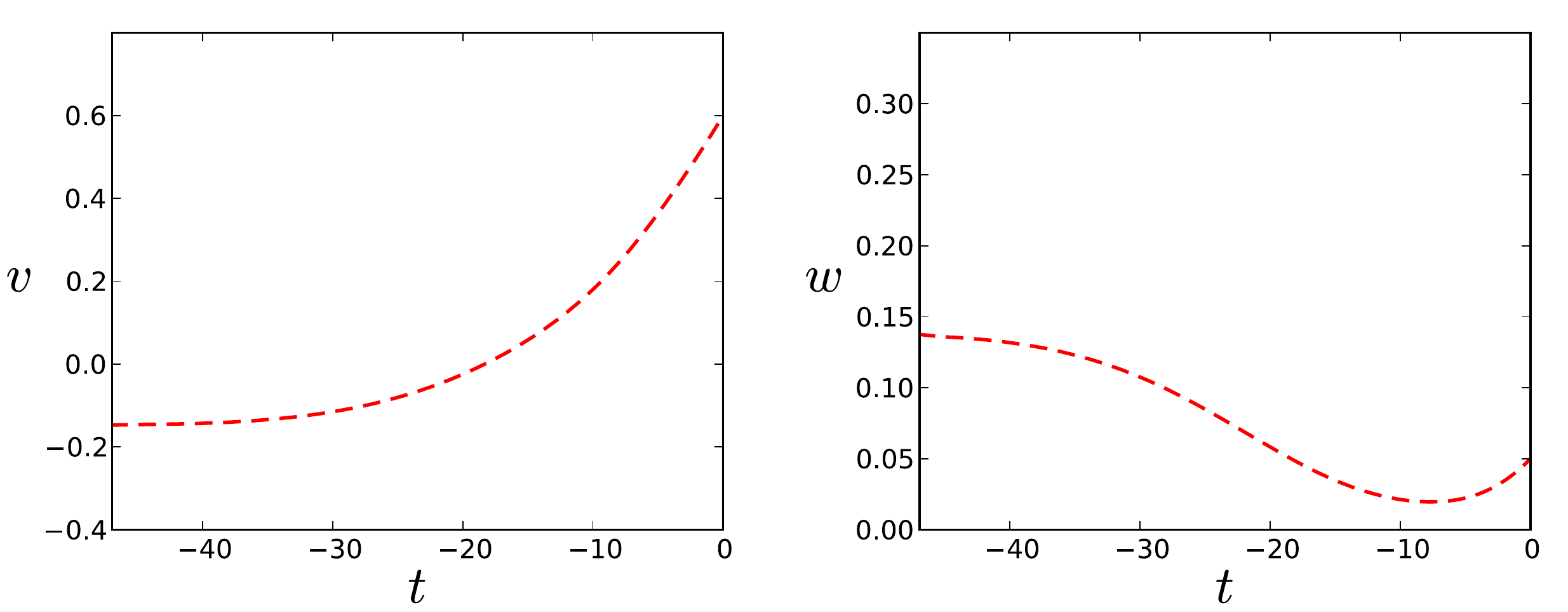}
  \caption{The most probable path (corresponding to the dashed red curve in Fig.~\ref{fig:type2a_pp2}) during initiation as a function of time along with several simulation trajectories (grey curves) for the Type II system.}
  \label{fig:mctraj_type2a}
\end{figure}

\begin{figure}[tbp]
  \centering
  \includegraphics[width=13cm]{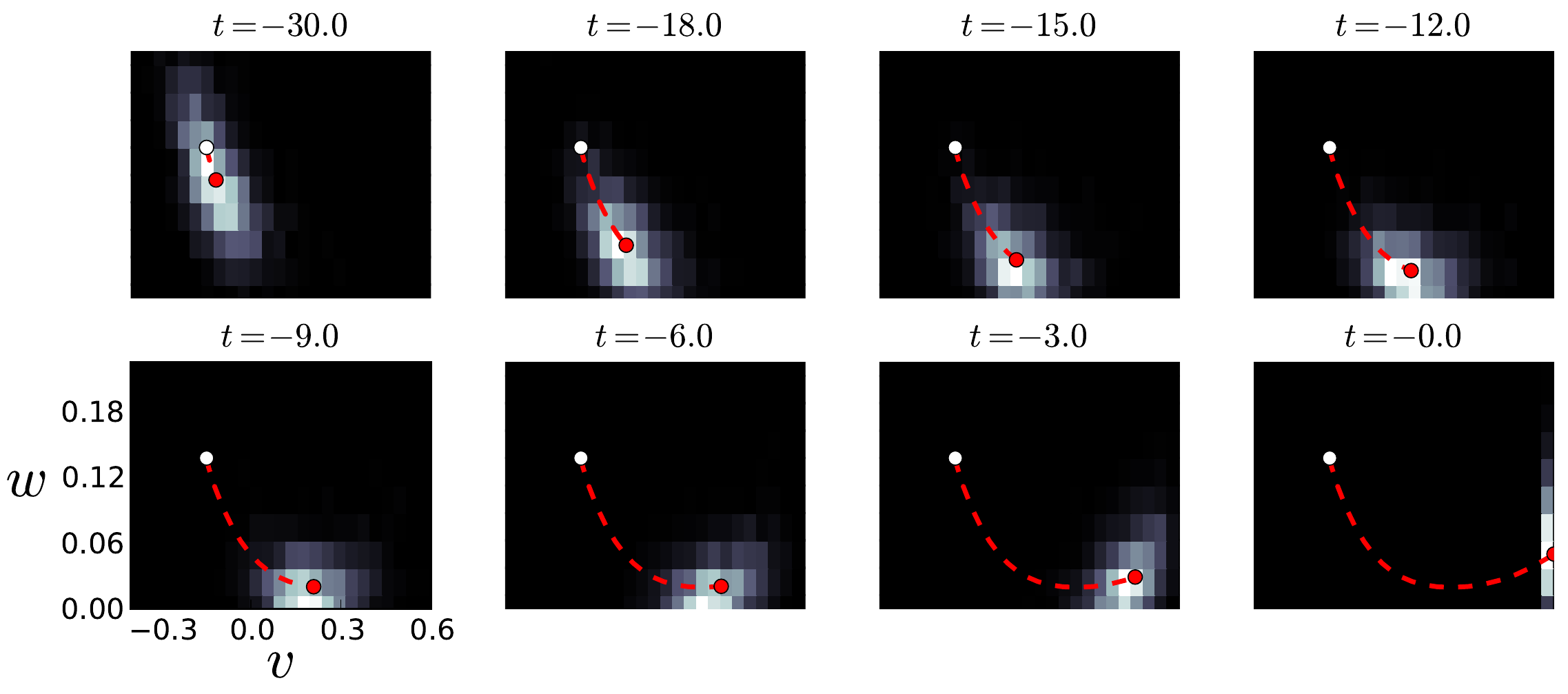}
  \caption{The path history $Q$ during the initiation phase of a Type II spontaneous action potential, computed from $10^{3}$ simulation trajectories. The dashed red curve is the MPP (the same curve as shown in Fig.~\ref{fig:mctraj_type2a}) up to time $t$. The stable fixed point is shown as a white circle.}
  \label{fig:mc_type2}
\end{figure}

\section{Discussion}
A spontaneous excitable event has two phases: the initiation phase and the excitation phase.
During the initiation phase, ion channel fluctuations push the voltage and \Kt~channel population to a threshold.
The MPP taken during the initiation phase is most likely to follow the most probable exit path (MPEP), which is very different from any deterministic trajectory.
The excitation phase takes the system from the threshold through a transient spike in voltage and ultimately back to the stable fixed point.
The MPP taken during the excitation phase follows (at least in part) closely to a deterministic trajectory.

For the basic type I system, there is a single MPP that characterizes the spontaneous action potential.
During the initiation phase, the MPP moves from the stable fixed point to the saddle along the MPEP (red curve Fig.~\ref{fig:type1a_pp2}).
During the excitation phase, the MPP follows the deterministic unstable manifold around the unstable fixed point back toward the stable fixed point (green curve Fig.~\ref{fig:type1a_pp2}).
The initiation phase for the type I system with bursting is similar (red curve Fig.~\ref{fig:type1b_qp}(a)).
Unlike the previous case, the excitation phase has multiple parts.
The excitation phase starts with the right branch of the unstable manifold until reaching the stable limit cycle (green curve Fig.~\ref{fig:type1b_pp1}(a)).
The number of oscillations around the unstable fixed point, and therefore the length of the burst, is not described by a MPP, but after fluctuations push the trajectory far enough away from the limit cycle, an MPP describes its approach back to the saddle (blue curve Fig.~\ref{fig:type1b_pp2}).
After returning to the saddle, the trajectory is most likely to return to the stable fixed point along the left branch of the unstable manifold (green curve Fig.~\ref{fig:type1b_pp1}(a)).

The analysis of the type II system is less tractable due to the lack of a saddle.
However, the presence of a caustic provides an effective energy barrier for initiation.
Even without the saddle, MPPs follow closely along a single path (dashed red curve Fig.~\ref{fig:type2a_qp}(a)) as they approach the energy barrier.
Hence, a single MPP characterizes the initiation phase similar to the type I case described above.
However, during the excitation phase, the amplitude of the spontaneous action potential is not described by a single MPP because there are many MPPs that have approximately equal likelihood (blue region in Fig.~\ref{fig:type2a_qp}(a)).

The biggest difference between type I and type II spontaneous action potentials is the behavior during the excitation phase.
Type I action potentials can exhibit voltage oscillations, especially when there is a stable limit cycle surrounding the unstable fixed point.
Type II action potentials on the other hand do not show this behavior.

Interestingly, the behavior of type I and type II spontaneous action potentials are similar during the initiation phase.
There are two ion channel species that contribute to spontaneous initiation.
The number of open channels determines the net current $\rf(v, m, n)$, and if the net current is increased for enough time, the voltage rises above threshold, generating an action potential.
There are two ways to increase the net current: by opening \Nat~channels or by closing \Kt~channels.

Clearly, the maximum increase in the net current from closing \Kt~channels occurs when all of the \Kt~channels are closed.
If we fix $m = 0$ to be constant, removing \Kt~channel fluctuations and dynamics, the deterministic system becomes
$$ v' = \rf(v, 0, Nx_{\infty}(v)). $$
The function $\rf(v, 0, Nx_{\infty}(v))$ has cubic-like shape.
For $I_{\app} < I_{*}$, there are three fixed points where $v' = \rf(v, 0, Nx_{\infty}(v)) = 0$: two stable separated by one unstable.
Only the third fixed point is above the threshold.
To generate an action potential, \Nat~channel fluctuations are required to increase the voltage from the first fixed point past the second.
At $I_{\app} = I_{*}$, the first two fixed points vanish.
For $I_{\app} > I_{*}$, only the third fixed point remains, which means that \Kt~channels alone are capable of initiating an action potential.
That is, once all of the \Kt~channels close, the voltage can deterministically increase above threshold.
Fig.~\ref{fig:type1a_pp1} shows the $I_{\app}<I_{*}$ case while Fig.~\ref{fig:type2a_pp1} shows the $I_{\app}>I_{*}$ case.

Recall that the parameter $\hgam = 1/(\epsilon M)$ controls the relative strength of \Nat~and \Kt~channel fluctuation.
As $\hgam$ decreases (equivalently $M$ increases with $\epsilon$ fixed) the \Kt~channel fluctuations become less significant than \Nat~channel fluctuations.
For $\hgam$ small enough, \Nat~channels provide the dominant contribution to spontaneous initiation.
It is natural to conclude then that for $\hgam$ large enough, \Kt~channels provide the dominant contribution to spontaneous initiation.
Indeed, this is the case when $I_{\app} > I_{*}$.
Notice that MPPs leading to action potentials drop below the $v$-nullcline in Fig.~\ref{fig:type2a_pp2}.
If instead we have $I_{\app} < I_{*}$ then \Kt~channels alone cannot induce an action potential, regardless of how large $\hgam$ is.

\vfill

\pagebreak

\appendix

\section{Ordered upwind method}
\label{sec:app_oum}
The ordered upwind method (OUM) is a finite difference method that approximates the quasipotential at a set of discrete grid points.
That is, it approximates the solution to the static Hamilton--Jacobi equation,
\begin{equation}
  \label{eq:3}
  \mathcal{H}(\vx, \nabla W) = 0,
\end{equation}
where ${\cal H}$ is given by \eqref{eq:64}.
The method is well known for solving static Hamilton-Jacobi equations \cite{sethian01a} and has been recently adapted for use in the stochastic setting for continuous Markov processes \cite{cameron12a}.
This adaptation takes advantage of a geometric minimum action formulation of the path integral \cite{heymann08a}.
However, the algorithm presented in Ref.~\cite{cameron12a} works only for continuous Markov processes and must be modified as follows.

The method exploits two facts about the Hamiltonian dynamical system: (i) the action $\act(t)$ is an increasing function of $t$ and satisfies a least action principle; and (ii) the quasipotential is related to the action on the ${\cal H} = 0$ MPPs by $W(\vx(t)) = S(t)$ (see Section \ref{sec:mpp}).
The solution surface is initially defined on the interior of a level curve of the quasipotential $W$ surrounding a stable fixed point.

For grid points in a small region containing the stable fixed point, the solution is known before hand.
The initially computed region is the interior of a level curve of $W$, the function we seek to compute.
That is, a curve $\vx(s)$ such that $W(\vx(s) = \text{Const}$.
The method can be described as sequentially computing the solution at grid points closest to those points previously computed 
As it grows outward, the outer boundary between computed and uncomputed grid points is approximately a level curve of $W$.

To organize the computation, grid points are categorized as follows.
The grid points within the computed region nearest to the boundary are labeled {\em accepted front}; all other points in the computed region are labeled {\em accepted}.
All grid points exterior to this region must be computed one at a time.

Any uncomputed grid point that is adjacent to an accepted front point is labeled a {\em considered} point.
At each step, a tentative value of $W$ is computed for each considered point using a finite diference approximation described below. 
All of the uncomputed grid points that are not considered points (those not close to the accepted front) are labeled {\em unconsidered}.
The finite difference formula, detailed below, requires two accepted front points, adjacent to each other, to update a considered point.  
For a given considered point, the finite difference is computed from each adjacent pair of accepted front points within a pre specified range, and the pair that results in the smallest value of $W$ is chosen as a tentative value. 

The pre specified range, call it $\gamma$, is necessary in order for information to propagate along characteristics.
If the Hamiltonian is very simple (quadratic in $\vp$ and from a system satisfying detailed balance) then the update will always come from an adjacent pair because characteristics follow the gradient of $W$.
This is not so for nonequilibrium steady states.
The characteristic leading to the considered point might not pass between two of the closest accepted points, it might pass between a pair of accepted front points that are farther away.
This happens when the angle between the characteristic and the level curve is small.
In particular, we know that this angle is small when characteristics are close to deterministic trajectories and $\norm{\vp}\ll 1$.

At the beginning of each iteration of the method, the considered point with the smallest tentative value of $W$ is chosen and becomes a new accepted front point.
Any of the previous accepted front points that are now on the interior become accepted points.
Unconsidered points that are adjacent to the new accepted front point become considered points, and tentative values of $W$ are computed for each new considered point.
Additionally, considered points that are within the pre specified range $\gamma$ of the new accepted front point have their tentative values recomputed.
Once all tentative values have been computed as necessary, the process repeats by choosing the considered point with the smallest tentative value.
The method stops when there are no more considered or unconsidered points (or when a specified maximal value of $W$ is reached).

We summarize the algorithm as follows
\begin{enumerate}
  \item[Step 1.] Find the considered point with the smallest value of $W$ (a good choice to keep a sorted list is the heap sort algorithm)
  \item[Step 2.] Relabel the chosen considered point as accepted front
  \item[Step 3.] Relabel any nearby accepted front point that has no adjacent considered point as accepted
  \item[Step 4.] Relabel nearby unconsidered points adjacent to the new accepted front point as considered
  \item[Step 5.] Compute a tentative value of $W$ for each considered points with in the range $\gamma$ of the new accepted front point as described above
  \item[Step 6.] Repeat until all grid points have been computed or a prespecified value of $W$ has been reached.
\end{enumerate}

The Gaussian approximation detailed in Section \ref{sec:linear-theory} can be used to initialize the method.  
The solution to the algebraic Ricatti equation \eqref{eq:88} allows us to generate a small elliptical boundary around the stable fixed point on which we specify initial data for the characteristics.
The solution to \eqref{eq:88} at $\vxa$ can be computed by setting $S = Z^{-1}$ and multiplying both sides of \eqref{eq:88} by $Z^{-1}$ to obtain the linear problem
\begin{equation}
  \label{eq:89}
  D + CS + SC^{T} = 0.
\end{equation}
Let $\delta \ll 1$ be the value of $W$ on the initial level curve.
The initial region is the set of grid points $\{ \vx_{n} : W(\vx_{n}) \leq \delta \} $ on the interior of the initial level curve, where
\begin{equation}
  \label{eq:59} 
  W(\vx_{n}) = \frac{1}{2}(\vx_{n} - \vxa )^{T}Z(\vxa)(\vx_{n} - \vxa ),
\end{equation}
and
\begin{equation}
  \label{eq:60}
  \nabla W(\vx_{n}) = Z(\vxa)(\vx_{n} - \vxa ).
\end{equation}
Characteristics converge to MPPs that start at the stable fixed point after taking the limit $\delta \to 0$.
For numerical solutions, we take $\delta \ll 1$ large enough to obtain a stable solution but small enough that quasipotential is accurate.
Initial data can also be specified near a stable limit cycle; we leave the details to Appendix \ref{sec:app_lc}.

\subsection{Variational finite difference}
Assume that the Hamiltonian is a convex function of $\vp$.
Compute the finite difference approximation at a point $\vx$ using the adjacent pair of accepted front points, $\vx_{1}$ and $\vx_{2}$, with corresponding $\vp_{1}$ and $\vp_{2}$ previously computed.
Define 
\begin{equation}
  \vx_{\theta} = \theta \vx_{1} + (1-\theta)\vx_{2}, \quad 0<\theta < 1.
\end{equation}
Likewise, let $\vp_{\theta} = \theta \vp_{1} + (1-\theta)\vp_{2}$ and $W_{\theta} = \theta W(\vx_{1}, \vp_{1}) + (1-\theta)W(\vx_{2}, \vp_{2})$.
Differentiating $W(\vx(t))$ with respect to time yields the differential relationship $dW = d\vx\cdot \vp$.
Hence, given $\theta$ and $\vp$, a finite difference formula is given by
\begin{equation}
  \label{eq:4}
      W(\vx) \approx  W_{\theta} + d\vx_{\theta}\cdot \vp.
\end{equation}
To specify $\theta$ and $\vp$, we use a variational finite difference formula \cite{heymann08a}, given by
\begin{equation}
  \label{eq:33}
  W(\vx) \approx \inf_{\theta\in (0, 1)}\left\{ W_{\theta} + \sup_{\stackrel{\vp\in \mathbf{R}^{2}}{ \mathcal{H}(\vp)=0}}\left(d\vx_{\theta}\cdot \vp - \mathcal{H}(\vp)\right)  \right\}.
\end{equation}
The outer minimization follows from the least action principle and can be numerically computed using standard minimization routines.
The inner maximization is the Legendre transform of the Hamiltonian, with the added constraint that $\mathcal{H}(\vp) = 0$.
Since the Legendre transform of the Hamiltonian is the Lagrangian, this constrains $\theta$ so that $d\vx_{\theta}$ is 'upwind' of $\vx$ in the direction of the MPP passing through $\vx$.
In other words, the finite differences move information along characteristics.

A numerical method to compute the maximizer, $\vp(d\vx_{\theta})$, is as follows.
Using a Lagrange multiplier $\mu$ we want to compute
\begin{equation}
  \label{eq:6}
  \sup_{\vp\in \mathbf{R}^{2}}\left[d\vx_{\theta} \cdot \vp - \mu^{2} \mathcal{H}(\vx_{\theta}, \vp)\right],\quad {\cal H}(\vx_{\theta}, \vp) = 0.
\end{equation}
It follows that the maximizer satisfies
\begin{equation}
  \label{eq:7}
  \mathcal{H}(\vx_{\theta}, \vp) = 0,\quad d\vx_{\theta} - \mu^{2} \nabla_{\vp}\mathcal{H}(\vx_{\theta}, \vp) = 0.
\end{equation}
Rewrite the second equation as
\begin{equation}
  \label{eq:8}
\nabla_{\vp}\mathcal{H}(\vx_{\theta}, \vp) = \lambda \frac{d\vx_{\theta}}{\norm{d\vx_{\theta}}},
\end{equation}
where $\lambda = \norm{d\vx_{\theta}}/\mu^{2}$.
Recall that $\nabla_{\vp}\mathcal{H}(\vp) = \vx'$ along characteristics.  
Hence, the Lagrange multiplier $\lambda$ has the interpretation $\lambda = \frac{ds}{dt}$, where $s$ is arclength along a characteristic.

If we replace ${\cal H}$ by its diffusion approximation by expanding ${\cal H}$ around $\vp = 0$ to second order in $\vp$, the exact solution to the variational problem \eqref{eq:7} can be calculated.
This idea can be used to derive an iterative numerical procedure, which is equivalent to Newton's method, by successively computing the solution to a locally valid approximation given by expanding ${\cal H}$ around $\vp = \vp_{n}$ to second order.
Given $\vx_{\theta}$, $d\vx_{\theta}$, and an initial guess $\vp_{0}$, the numerical approximation $\vp_{n+1}$, $\lambda_{n+1}$ satisfies
\begin{equation}
  \label{eq:9}
  \vp_{n+1} = \vp_{n} + \mathcal{H}_{\vp \vp}^{-1}\left[\lambda_{n}\frac{d\vx_{\theta}}{\norm{d\vx_{\theta}}} - \mathcal{H}_{\vp}\right],
\end{equation}
where
\begin{equation}
  \label{eq:10}
   \mathcal{H}_{\vp} =  \nabla_{\vp}\mathcal{H}(\vx_{\theta}, \vp_{n}), \quad \{\mathcal{H}_{\vp\vp}\}_{ij} = \frac{\partial^{2}}{\partial p_{i} \partial p_{j}}[\mathcal{H}(\vx_{\theta}, \vp_{n})].
\end{equation}
 The Lagrange multiplier is given by
\begin{equation}
  \label{eq:11}
  \lambda_{n} =
  \begin{cases}
\norm{d\vx_{\theta}}\sqrt{\alpha}, & \alpha \geq 0\\
0 & \alpha < 0
  \end{cases},\quad  \alpha = \frac{\mathcal{H}_{\vp}^{T}\mathcal{H}_{\vp\vp}^{-1}\mathcal{H}_{\vp} - 2\mathcal{H}(\vx_{\theta}, \vp_{n})}{d \vx_{\theta}^{T} \mathcal{H}_{\vp\vp}^{-1}d \vx_{\theta}}.
\end{equation}

\subsection{Initial data for stable limit cycles}
\label{sec:app_lc}
Suppose the deterministic system has a stable limit cycle $\vx_{c}(t)$. 
Introduce the orthogonal coordinate system $(s, r)$ with $0\leq s \leq L$ an arclength parameterization of the limit cycle and $r$ the signed distance from the limit cycle with respect to the direction outward normal to the limit cycle.
The deterministic dynamics,
\begin{equation}
  \label{eq:104}
  \vx' = \nabla_{\vp}\mathcal{H}(\vx, 0) \equiv \mathbf{F}(s, r),
\end{equation}
can be expanded around the limit cycle with
\begin{equation}
  \label{eq:105}
  \mathbf{F}(\vx_{c}(s) + r\hat{\bm{\eta}}) \sim B(s)\hat{\bm{\tau}} + a_{0}(s)r\hat{\bm{\eta}},
\end{equation}
where $\hat{\bm{\tau}}$ and $\hat{\bm{\eta}}$ are the unit tangent and normal vectors, respectively, and
\begin{equation}
  \label{eq:106}
   B(s)\hat{\bm{\tau}}  = \mathbf{F}(x_{c}), \quad
  a_{0}(s)  = \lim_{r\to 0}\pd{}{r}\hat{\bm{\eta}}\cdot \mathbf{F}(\vx_{c} + r\hat{\bm{\eta}})= \hat{\bm{\eta}}\cdot\left(\pd{\mathbf{F}(\vx_{c})}{(v, w)}\cdot\hat{\bm{\eta}}\right).
\end{equation}
The vector tangent to the limit cycle is $\bm{\tau} \equiv \mathbf{F}(\vx_{c})$, and we define the {\em unit} tangent vector so that $\hat{\bm{\tau}} = \nabla s$.
The normalization factor, $\norm{\bm{\tau}(s)} = B(s)$, is the speed of the deterministic trajectory along the limit cycle.
The unit outward normal vector is then given by $\hat{\bm{\eta}}  = (-\hat{\tau}_{2}, \hat{\tau}_{1}) = \nabla r$.

We want to approximate the quasipotential in a neighborhood of the limit cycle.
Expanding around $r = 0$ yields
\begin{equation}
  \label{eq:107}
  W(s, r) \sim \frac{r^{2}}{2}\pdd{W}{r}.
\end{equation}
Let $\phi(s) = \pdd{}{r}W(s, 0)$.
In order to get initial data for characteristics originating from a stable limit cycle, one need only compute $\phi(s)$.
Substituting \eqref{eq:107} and \eqref{eq:105} into $\widetilde{H}(x, p) = 0$, where $\widetilde{H}$ is given by \eqref{eq:90}, and taking the limit $r\to 0$ yields
\begin{equation}
  \label{eq:108}
 \frac{1}{2}B(s)\phi' + a_{0}(s)\phi +  D_{c}(s)\phi^{2} = 0, \quad \phi(0) = \phi(L),
\end{equation}
where
\begin{equation}
  \label{eq:109}
  D_{c}(s) \equiv \sum_{i,j}D_{ij}(s)\pd{r}{x_{i}}\pd{r}{x_{j}},\quad D_{ij}(s) \equiv \frac{\partial^{2}}{\partial p_{i} \partial p_{j}}\mathcal{H}(x_{c}(s), 0).
\end{equation}
A Gaussian approximation near the limit cycle means that $D_{ij}$ is the diffusion tensor.
The above equation \eqref{eq:108} was derived in Ref.~\cite{schuss10a} in the context of an exit over a {\em characteristic boundary}, i.e., when the separatrix is a limit cycle.

Initial data is specified along the curve $(s, r_{0}(s))$ where $W(s, 0) = \delta$ is constant (i.e., a level curve of the quasipotential).
It follows that $r_{0}(s) = \pm \sqrt{\frac{2\delta}{\phi(s)}}$, and given a periodic solution to \eqref{eq:108}, the initial data is given by
\begin{equation}
  \label{eq:110}
 \vx_{0}(s) = \vx_{c}(s) + r_{0}(s)\hat{\bm{\eta}}(s), \quad
 \vp_{0}(s) \sim r_{0}(s)\phi(s)\hat{\bm{\eta}}(s) + O(r_{0}^{2}).
\end{equation}
Equation \eqref{eq:109} can be converted to a linear equation by setting $\phi(s) = 1/\alpha(s)$ to get
\begin{equation}
  \label{eq:111}
  \frac{1}{2}B(s)\alpha' - a_{0}(s)\alpha = D_{c}(s), \quad \alpha(0) = \alpha(L).
\end{equation}
The solution is
\begin{equation}
  \label{eq:112}
  \alpha(s) = e^{\psi(s)}\left[R(s) + \frac{e^{\psi(L)}}{1 - e^{\psi(L)}}R(L) \right],
\end{equation}
where
\begin{equation}
  \label{eq:113}
  \psi(s) \equiv 2\int_{0}^{s}\frac{a_{0}(u)}{B(u)}du, \quad R(s) \equiv 2\int_{0}^{s}e^{-\psi(u)}\frac{D(u)}{B(u)}du.
\end{equation}
Hence,
\begin{equation}
  \label{eq:114}
  \phi(s) = e^{-\psi(s)}\left[R(s) + \frac{e^{\psi(L)}}{1 - e^{\psi(L)}}R(L) \right]^{-1}.
\end{equation}
Often the limit cycle must be computed numerically.  
In this case, it is more stable and efficient to solve \eqref{eq:108} numerically using an implicit finite difference scheme than to evaluate \eqref{eq:114} with numerical quadrature.

\section{Monte Carlo simulation algorithm}
\label{sec:monte-carlo-simul}
Monte-Carlo simulations are generated using an extension of the algorithm presented in \cite{keener11a}.  Instead of using the Gillespie algorithm as in \cite{keener11a}, we use the next reaction method along the lines of \cite{bokes13a}.
The algorithm is exact in the sense that the transition times can be approximated to any desired precision.
The simulations were coded in C (using the GNU Scientific Library for random number generators) and carried out in Python, using the Scipy package.
In between each jump in the number of open channels, the voltage is evolved according to the deterministic dynamics 
\begin{equation}
  \label{eq:13}
  \frac{dv}{dt} = \frac{n}{N}f_{\Na}(v) + \frac{m}{M}f_{\K}(v) + f_{\rl}(v) + I_{\app},
\end{equation}
The solution provides the relationship between voltage and time,
\begin{equation}
  \label{eq:14}
  v(t) = \left(v(t_{0}) - \frac{c_{2}}{c_{1}}\right)e^{-c_{1}(t-t_{0})} + \frac{c_{2}}{c_{1}},
\end{equation}
where
\begin{align}
  \label{eq:15}
  c_{1} &= \frac{n}{N}g_{\Na} + \frac{m}{M}g_{\K} + g_{\rl}, \\
  c_{2} &= \frac{n}{N}g_{\Na}v_{\Na} + \frac{m}{M}g_{\K}v_{\K} + g_{\rl}v_{\rl} + I_{\app}.
\end{align}
To compute the next jump time, we compute four random jump times for each of the four possible transitions: $n\to n\pm 1$ and $m \to m \pm 1$. 
Each transition time is distributed according to
\begin{gather}
  \label{eq:16}
  W^{-}_{\Na}(t) =1 - e^{-\beta_{\Na}n (t - t_{0})}, \quad
  W^{+}_{\Na}(t) = 1- \explr{-\beta_{\Na}\int_{t_{0}}^{t}\Omega^{+}_{\Na}(v(\tau))d\tau},\\
  W^{\pm}_{\K}(t) = 1- \explr{-\beta_{\K}\int_{t_{0}}^{t}\Omega^{\pm}_{\K}(v(\tau))d\tau}.
\end{gather}
After integrating the voltage dependent transition rates we obtain for $(i=+, j = {\rm Na})$ and $(i=\pm, j = {\rm K})$,
\begin{equation}
  \label{eq:17}
  \int_{t_{0}}^{t}\Omega^{i}_{j}(v(\tau))d\tau = \frac{1}{c_{1}}\Omega^{i}_{j}(\frac{c_{2}}{c_{1}})(E_{\rm i}(z^{i}_{j}e^{-c_{1}(t-t_{0})}) - E_{\rm i}(z^{i}_{j})),
\end{equation}
where
\begin{equation}
  \label{eq:18}
  z^{+}_{\Na} =  4\gamma_{\Na}\left(v(t_{0}) - \frac{c_{2}}{c_{1}}\right), \quad
  z^{\pm}_{\K} = \pm \gamma_{\K}\left(v(t_{0}) - \frac{c_{2}}{c_{1}}\right), 
\end{equation}
and $E_{\rm i}$ is the exponential integral function defined as the Cauchy principal value integral,
\begin{equation}
  \label{eq:19}
    E_{\rm i}(x) = \int_{-\infty}^{x}t^{-1}e^{t}dt, \quad x\neq 0.
\end{equation}
Denote the jump times by $t^{i}_{j}$, $i = \pm$ and $j = {\rm Na}, {\rm K}$, and let $U$ be a uniform random variable.
The jump times are given by the solution to $W^{i}_{j}(t^{i}_{j}) = U$.
There is one voltage independent jump time,
\begin{equation}
  \label{eq:21}
  t^{-}_{\Na} = -\frac{\log(U)}{n\beta_{\Na}}.  
\end{equation}
Because three of the transition rates depend on voltage, and therefore time, the distributions for the jump times are not explicitly invertible.
Hence, the next jump times are given implicitly by
\begin{gather}
\nonumber
\frac{1}{c_{1}}\Omega^{+}_{\Na}(\frac{c_{2}}{c_{1}})(E_{\rm i}(z^{+}_{\Na}e^{-c_{1}(t^{+}_{\Na}-t_{0})}) - E_{\rm i}(z^{+}_{\Na})) = -\frac{\log(U)}{\beta_{\Na}}, \\
  \label{eq:22}
\frac{1}{c_{1}}\Omega^{\pm}_{\K}(\frac{c_{2}}{c_{1}})(E_{\rm i}(z^{\pm}_{\K}e^{-c_{1}(t^{\pm}_{\K}-t_{0})}) - E_{\rm i}(z^{\pm}_{\K})) =-\frac{\log(U)}{\beta_{\K}}.
\end{gather}
To generate the voltage dependent jump times, a root finding algorithm is applied to \eqref{eq:22} with a tolerance of $ 10^{-8}$.  
Once all four transition times have been computed, the next transition time is $t^{i_{*}}_{j_{*}} = \min_{i=\pm,j=\Na,\K}\{t_{j}^{i}\}$. 
The global time is updated with $t \leftarrow t + t^{i_{*}}_{j_{*}}$.  The state is updated with $v \leftarrow v(t^{i_{*}}_{j_{*}})$ (where $v(t)$ is given by \eqref{eq:14} with $t_{0}$ the time of the previous jump), $n \leftarrow n + i_{*}$ if $j_{*} = {\rm Na}$, and $m \leftarrow m + i_{*}$ if $j_{*} = {\rm K}$.

\section{Large deviation principle}
\label{sec:app_path}
\newcommand{\ts}{N}
The derivation of the large deviation principle is based on the idea of {\em large deviations from an averaged system} (see Refs.~\cite{freidlin12a,kifer09a}).
The large deviation principle has been rigorously established for the case where the slow process is deterministic, which corresponds to the $M\to\infty$ limit where the \Kt~channel conductance is deterministic.
However, with slight modifications, their result can be applied to the stochastic ML model.
Let the random process $N(t/\epsilon)$ represent the number of open \Nat~channels.
For fixed $\vx$ it is a finite Markov chain with transition rate matrix $\mathbb{L}_{\Na}$ (the infinitesimal generator is the adjoint matrix).
The generalization of \eqref{eq:72} to the path distribution is
\begin{equation}
  \label{eq:152}
      \int_{t_{0}}^{t}\mathcal{H}(\vx(s), \vp(s))ds  =  
 \lim_{\epsilon\to 0}\epsilon \log \ave{\explr{\frac{1}{\epsilon}\int_{0}^{t}H(N(s/\epsilon), \vx(s), \vp(s))ds}},
\end{equation}
where
\begin{equation}
  \label{eq:153}
  H(N(s/\epsilon), \vx(s), \vp(s)) \equiv \pv(s) \rf(\vx(s), N(s/\epsilon)) + h(\vx(s), \pw(s)).
\end{equation}
The function $h$ is given by \eqref{eq:45} and limits to $h \to \pw(w_{\infty}(v) - w)/\tau_{w}(v)$ as $M \to \infty$.
We assume that \eqref{eq:152} implies a large deviation principle and the approximation \eqref{eq:154}.
A simpler formula for $\mathcal{H}$ than \eqref{eq:152} is obtained using the separation of timescales.
Change variables to the fast timescale $T = t/\epsilon$ to get
\begin{equation*}
    \epsilon \int_{0}^{T}\mathcal{H}(\vx(t), \vp(t))dT' \sim  
 \epsilon \log \ave{\explr{\int_{0}^{T}H(\ts(T'), \vx(t), \vp(t))dT'}},
\end{equation*}
For fixed $t$ we have that $T\to\infty$ as $\epsilon\to 0$.
Divide through by $\epsilon T$, set $\vx(t) = \vx$ and $\vp(t) = \vp$, and take the limit $\epsilon\to 0$ with $t$ fixed to get
\begin{equation*}
  \mathcal{H}(\vx, \vp) = \lim_{T\to\infty}\frac{1}{T} \log \ave{\explr{\int_{0}^{T}H(\ts(T'), \vx, \vp)dT'}}.
\end{equation*}
The above mean assumes that $\vx$ is fixed constant so that
\begin{equation}
\label{eq:155}
  \mathcal{H}(\vx,\vp) = \lim_{T\to\infty}\frac{1}{t} \log(\sum_{s'}U_{s, s'}(t)),
\end{equation}
where $U(t) = \exp\{t\mathbb{M}^{T}\}$ and $\mathbb{M} = \mathbb{L}_{\Na} + \diag{H(n, \vx, \vp)}$.
One can show \cite{freidlin12a,kifer09a} using the properties of positive semigroups that \eqref{eq:155} converges to $\lambda$, uniquely defined as the principal eigenvalue of $\mathbb{M}$;  it is real and simple, it is greater than the real part of the remaining eigenvalues, and its eigenvector $r$ is strictly positive.
The {\em Perron--Frobenius Theorem} (see Lemma \ref{lem:PF}) guarantees that the principal eigenvalue exists with $r>0$ for all $\vp$.

\section{Establishing path equivalence}
\label{sec:app_pathequiv}
Consider the $(n+1)\times (n+1)$ matrix $M(\vp) = A + D(\vp)$, where $D$ is a diagonal matrix whose elements are $C^{\infty}(\mathbb{R})$ in $\vp$, with $D(0) = 0$ and $D(\vp) \neq 0$ for $\vp \neq 0$.
Assume that the elements of $A$ and $D$ are bounded continuously differentiable functions of $\vx$, mapping $\mathcal{D} \subset \mathbb{R}^{2} \to \mathbb{R}^{2}$.
Assume that $A$ is an irreducible transition rate matrix.
That is, the diagonal elements are negative, the off diagonal elements of $A$ are nonnegative, and $\sum_{i=1}^{n}A_{ij} = 0$.
\begin{lemma}
\label{lem:PF}
  For fixed $\vp$, the following statements hold regarding the matrix $M = A + D(\vp)$ 
  \begin{enumerate}
  \item[(i)] There is exactly one positive eigenvector $q$
  \item[(ii)] The eigenvalue $\mathcal{H}$ corresponding to $q$ is real and simple
  \item[(iii)] $\mathcal{H}$ is greater than the real part of the remaining eigenvalues
  \end{enumerate}
\end{lemma}
\begin{proof}
  By assumption on $A$ and since $D$ is diagonal, there exists a scalar $\kappa$ such that the matrix $U = A + D(\vp) + \kappa I$ is nonnegative with positive diagonal entries.
Note that $U$ is irreducible if $A$ is irreducible.
  It follows from the Perron--Frobeneous Theorem that $U$ has exactly one positive eigenvector $q$ with a real, simple eigenvalue $\mu$ that is greater than the real part of all the remaining eigenvalues.
  Let $\mu_{j}, q_{j}$ be an eigenpair of $U$, with $\mu_{j}\neq \mu$.
  Setting $\mu_{j} = \kappa + \lambda_{j}$ we have that
  \begin{equation*}
    [A + D(\vp)]q_{j} + \kappa q_{j} = \kappa q_{j} + \lambda_{j} q_{j}.
  \end{equation*}
  It follows that $\lambda_{j}, q_{j}$ is an eigenpair of $A + D(\vp)$, which establishes (i) and (ii).
  Moreover, we have that
  \begin{equation*}
    \mu > \Re(\mu_{j}) \Rightarrow \kappa + \mathcal{H} > \Re(\kappa + \lambda_{j})\Rightarrow \mathcal{H} >  \Re(\lambda_{j}),
  \end{equation*}
which proves (iii).
\end{proof}
Define the Hamiltonian $\mathcal{H}(\vp)$ as the Perron--Frobeneous eigenvalue of $M$.
The problem with this definition is that $\mathcal{H}$ is only implicitly defined as a root of a characteristic polynomial.
Therefore, an explicit formula for the Hamiltonian is only possible in special cases.
For practical problems we need a general way to write the Hamiltonian with an explicit formula.
Consider the alternative definition of the Hamiltonian
\begin{equation}
  \label{eq:25}
  \widehat{\mathcal{H}}(\vp) = \frac{1}{\tau} \det(A + D(\vp)),
\end{equation}
where $\tau \neq 0$ is a timescale independent of $\vp$.
It immediately follows that if $\mathcal{H}(\vp) = 0$ then $\widehat{\mathcal{H}}(\vp) = 0$.
If one is interested only in an approximation of the stationary density or mean first exit times, using $\widehat{\cal H}$ is completely equivalent to using $\mathcal{H}$.
However, the connection to the Lagrangian from the path integral formulation and therefore to MPPs is less clear because $\widehat{\mathcal{H}}$ is generally not a convex function of $\vp$, which complicates the Legendre transform.

The explicit Hamiltonian \eqref{eq:25} can be used to compute the stationary density and mean exit times, but it does not necessarily yield the same characteristics as ${\cal H}$.
The four dimensional solution surface is parameterized with $(\vx(t, \theta), \vp(t, \theta))$, where $\theta$ parameterizes the initial data.
Curves of constant $\theta$ are characteristics.
Both Hamiltonians define the same solution surface given the same initial data, but the parameterization with respect to $t$ need not be the same.

One can show that characteristics of $\widehat{\mathcal{H}}$ generate the same curves of constant $\theta$ in $(\vx, \vp)$ as those of $\mathcal{H}$, and that the two differ only by a timescale.
That is, curves of contant $t$ may differ.
In other words, the characteristic projections generated by $\widehat{\cal H}$ are the same as those generated by $\mathcal{H}$, with different time parameterizations.
\begin{theorem}
  Given $\mathcal{H}$ and $\widehat{\mathcal{H}}$ as defined above, define the two dynamical systems
  \begin{equation}
    \label{eq:145}
    \frac{d\vx}{dt} = \nabla_{\vp}\mathcal{H}(\vx, \vp),\quad \frac{d\vp}{dt} = -\nabla_{\vx}\mathcal{H}(\vx, \vp),
  \end{equation}
  and
  \begin{equation}
    \label{eq:146}
    \frac{d\hat{\vx}}{ds} = \nabla_{\vp}\widehat{\mathcal{H}}(\hat{\vx}, \hat{\vp}),\quad \frac{d\hat{\vp}}{ds} = -\nabla_{\vx}\widehat{\mathcal{H}}(\hat{\vx}, \hat{\vp}),
  \end{equation}
  with initial conditions $\vx(0) = \hat{\vx}(0) = \vx_{0}$,  $\vp(0) = \hat{\vp}(0) =  \vp_{0}$, chosen so that $\mathcal{H}(\vx_{0}, \vp_{0}) = \widehat{\mathcal{H}}(\vx_{0}, \vp_{0}) = 0$.
  Assume that each system has a unique bounded and continuous solution for $t\in (0, T)$.
  There exists a continuous one to one mapping $\zeta : (0, \hat{T}) \to (0, T)$ such that
  \begin{equation}
    \label{eq:136}
    (\hat{\vx}(s), \hat{\vp}(s)) = (\vx(\zeta(s)), \vp(\zeta(s))).
  \end{equation}
\end{theorem}
\begin{proof}
Write the eigenvalues of $M$ as $\lambda_{i}$, $i=1,\cdots, n$ and set $\lambda_{1} = \mathcal{H}$.
Write the characteristic equation as
\begin{equation}
  \label{eq:26}
  \Omega(\lambda) \equiv a_{n}\lambda^{n} + a_{n-1}\lambda^{n-1}+ \cdots + a_{1}\lambda + a_{0} = 0,
\end{equation}
where $a_{0} = \det(A+D(\vp)) = \tau \widehat{\cal H}$.
By assumption on the elements of $A$ and $D$, the coefficients are continuously differentiable functions of $\vp$ and $\vx$.
Substituting $\lambda = \mathcal{H}$ into \eqref{eq:26} and differentiating yields
\begin{equation}
  \label{eq:27}
  \Omega'(\mathcal{H})\nabla_{\vp}\mathcal{H} = -\sum_{j=0}^{n}\mathcal{H}^{j}\nabla_{\vp}a_{j},\quad
  \Omega'(\mathcal{H})\nabla_{\vx}\mathcal{H} = -\sum_{j=0}^{n}\mathcal{H}^{j}\nabla_{\vx}a_{j},
\end{equation}
where
\begin{equation}
  \label{eq:144}
  \Omega'(\mathcal{H}) = \sum_{j=1}^{n}ja_{j}\mathcal{H}^{j-1}.
\end{equation}
By assumption, $\mathcal{H}(\vx(t), \vp(t)) = 0$.
Notice that
\begin{equation}
\Omega'(0) =   a_{1}(\vp) = \prod_{j=2}^{n}\lambda_{j}(\vp).
\end{equation}
Furthermore, since $\mathcal{H}$ is a simple eigenvalue and $\mathcal{H} > \Re(\lambda_{i})$, for $i = 2, \cdots, n$, the remaining eigenvalues must be nonzero, and it follows that $\Omega'(0) \neq 0$.

Substituting $\mathcal{H} = 0$ into \eqref{eq:27} and yields
\begin{equation}
   \frac{a_{1}}{\tau}\pd{\mathcal{H}}{p_{i}} = -\pd{\widehat{\cal H}}{p_{i}},\quad
   \frac{a_{1}}{\tau}\pd{\mathcal{H}}{x_{i}} = -\pd{\widehat{\cal H}}{x_{i}},
\end{equation}
We are free to choose $\tau$ so that $a_{1}(\vx, \vp)/\tau$ is bounded, continuous and positive for all $\vx \in \mathcal{D}$.
Setting $\zeta'(s) = a_{1}/\tau$ with $\zeta(0) = 0$, we have that $\zeta(s)$, $s\geq0$, is an increasing one-to-one function.
Hence, if $\vx(t), \vp(t)$ is a solution to \eqref{eq:145} then $\vx(\zeta(s)), \vp(\zeta(s))$ is a solution to \eqref{eq:146}. 
\end{proof}

\section{The pre exponential factor calculation}
\label{sec:prefactor}
Collecting second order terms in the WKB expansion and applying a solvability condition yields the prefactor equation,
\begin{multline}
  \label{eq:150}
   \pd{k}{v}\sum_{n}\rf(\vx, n)l(n | \vx)r(n | \vx) + \pd{k}{w}\pd{h}{\pw}  + \frac{1}{2}k\pdd{W}{w}\pdd{h}{\pw}\\
 + k\sum_{n}l(n | \vx)\left[\pd{}{v}(\rf(\vx, n)r(n | \vx)) + \pd{}{w}\left(\pd{h}{\pw}r(n | \vx)\right)\right]  = 0,
\end{multline}
where the left eigenvector satisfies
\begin{equation}
  \label{eq:118}
  \left[\frac{1}{\hgamna}\mathbb{L}^{*}_{\Na} + \pv(\vx) \rf(\vx, n) + h(\vx, \pw(\vx))\right]l(n | \vx) = 0.
\end{equation}
Along characteristics, one can show that
\begin{equation}
  \label{eq:149}
  \frac{dv}{dt} = \sum_{n}\rf(\vx, n)l(n | \vx)r(n | \vx), \quad \frac{dw}{dt} = \pd{h}{\pw}.
\end{equation}
It follows from \eqref{eq:150} that along characteristics, the pre exponential factor satisfies
\begin{equation}
  \label{eq:56}
\begin{split}
  \frac{dk}{dt} &= k\left [ \sum_{n}l(n | \vx)\pd{}{v}(\rf(\vx, n)r(n | \vx)) \right. \\
& \qquad \qquad \left.  + \pd{h}{\pw}\sum_{n}l(n | \vx)\pd{r}{w}(n | \vx) + \frac{\partial^{2}h}{\partial\pw\partial w} + \frac{1}{2}\pdd{W}{w}\pdd{h}{\pw}\right ].
\end{split}
\end{equation}
Note that \eqref{eq:56} requires the Hessian matrix \eqref{eq:58}, which satisfies \eqref{eq:38} on characteristics.

Finally to calculate the left eigenvector, we use the anzatz $l(n | \vx) = C^{n}$ for some yet to be determined constant $C$.
After substituting the anzatz into \eqref{eq:118} we obtain
\begin{equation}
  \label{eq:119}
  n\left[-C^{2}a_{\Na} + C\left(\frac{\hgamna\pv}{N}f_{\Na} a_{\Na} - 1\right) + 1\right] + NC\left[a_{\Na}C + \frac{\hgamna}{N}\left(\pv g + h\right) - a_{\Na}\right] = 0.
\end{equation}
Setting the $n$ independent term to zero reveals $C = \frac{A}{a_{\Na}}$.
The remaining $n$ dependent term in \eqref{eq:119} is zero since $\mathcal{H}(\vx, \vp) = 0$.
Normalizing $l(n | \vx)$ so that
\begin{equation}
  \label{eq:121}
  \mathcal{N}\sum_{n=0}^{N}l(n | \vx)r(n | \vx) = \mathcal{N}\sum_{n=0}^{N}\binom{N}{n}\Lambda^{n}(1-\Lambda)^{N-n}\left(\frac{A}{a_{\Na}}\right)^{n} = 1
\end{equation}
yields
\begin{equation}
  \label{eq:116}
  l(n | \vx) = \left(\frac{1 + A}{1 + A^{2}/a_{\Na}}\right)^{N}\left(\frac{A}{a_{\Na}}\right)^{n}.
\end{equation}
A lengthy but straightforward calculation, using \eqref{eq:67} and \eqref{eq:116}, shows that
\begin{equation}
  \sum_{n = 0}^{N}l(n | \vx)\pd{r}{w}(n | \vx) = \frac{N\pd{A}{w}(A - a_{\Na})}{(1 + A)(A^{2} + a_{\Na})},
\end{equation}
and
\begin{equation}
\nonumber
\begin{split}
  &\sum_{n = 0}^{N}l(n | \vx)\pd{}{v}(\rf(\vx, n)r(n | \vx)) \\
   &\qquad  =  N\pd{A}{v}\left[\frac{A\left(g - \Lambda f_{\Na}\right)}{A^{2} + a_{\Na}}
   + \frac{f_{\Na}A(A^{2} +\frac{a_{\Na}}{N})}{(A^{2} + a_{\Na})^{2}} - (1-\Lambda)g\right] \\
&\qquad \qquad +  \pd{g}{v} + \frac{f_{\Na}'A^{2}}{A^{2} + a_{\Na}}.
\end{split}
\end{equation}
Using \eqref{eq:97} we have that
\begin{equation}
  \label{eq:94} 
  \pd{A}{w} = -\frac{1}{N}\left[\pv \pd{g}{w} + \frac{\partial^{2}W}{\partial v\partial w} g + \pd{h}{w} + \pdd{W}{w}\pd{h}{\pw}\right],
\end{equation}
and
\begin{equation}
  \label{eq:95}
  \pd{A}{v} = a_{\Na}'  - \frac{1}{N}\left[\pv \pd{g}{v} + \pdd{W}{v} g + \pd{h}{v} + \frac{\partial^{2}W}{\partial w \partial v}\pd{h}{\pw}\right].
\end{equation}

\section{Parameter values}
\label{sec:appendix_params}
\subsection{Type I}
\label{sec:type-i}
$v_{\Na} = 1$, $g_{\Na}= 1$, $v_{\mathrm{K}}= -0.7$, $g_{\mathrm{K}}= 2$, $v_{\rl}= -0.5$, $g_{\rl}= 0.5$, $\beta_{\K} = 0.17$, $I_{\rm app}= 0$, $\gamma_{\Na} = 2.5$, $\kappa_{\Na}= 0.025$, $\gamma_{\K}= -3.45$, $\kappa_{\K}= 0.76$, $M = 200$, $N = 1$.

\subsection{Type I with bursting}
\label{sec:type-i-with}

$v_{\Na} = 1.15$, $g_{\Na}= 1$, $v_{\mathrm{K}}= -0.55$, $g_{\mathrm{K}}= 2$, $v_{\rl}= -0.35$, $g_{\rl}= 0.5$, $\beta_{\K} = 0.25$, $I_{\rm app}= 0.01$, 
$\gamma_{\Na} = 2.27$, $\kappa_{\Na}= -0.32$, $\gamma_{\K}= -10$, $\kappa_{\K}= 1.78$, $M = 200$, $N = 3$.

\subsection{Type II}
\label{sec:type-ii}
$v_{\Na} = 3.7$, $g_{\Na}= 0.22$, $v_{\mathrm{K}}= -0.9$, $g_{\mathrm{K}}= 0.4$, $v_{\rl}= -0.36$, $g_{\rl}= 0.1$, $\beta_{\K} = 0.04$, $I_{\rm app}= 0.06$, 
$\gamma_{\Na} = 1.22$, $\kappa_{\Na}= -1.188$, $\gamma_{\K}= -0.8$, $\kappa_{\K}= 0.8$, $M = 40$, $N = 40$.


\end{document}